\documentclass[final,3p,times]{elsarticle}
\usepackage{graphicx}      % include this line if your document contains figures
\makeatletter
\let\old@ssect\@ssect % Store how ifacconf defines \@ssect
\makeatother

\usepackage{amssymb}
\usepackage{amsfonts}
\usepackage{mathrsfs}
\usepackage{amsmath}
\usepackage{textcomp}
\usepackage{float}
\usepackage[dvipsnames]{xcolor}
\usepackage{multirow}
\usepackage{csquotes}
\usepackage{epstopdf} 
\usepackage{gensymb}
\usepackage{breqn}
\usepackage{siunitx}
\usepackage{hyperref} 
\hypersetup{colorlinks=true,colorlinks,linkcolor={blue},citecolor={blue},urlcolor={red}} 
\makeatletter
\def\endfigure{\end@float}
\def\endtable{\end@float}
\makeatother

\usepackage{subcaption}

\usepackage{booktabs, caption, nicematrix}
\usepackage[thicklines]{cancel}
\usepackage{amsthm}
\newtheoremstyle{boldtheorem}
  {\topsep}   % Space above
  {\topsep}   % Space below
  {\normalfont}%{\bfseries} % Body font (bold)
  {}          % Indent amount (empty means no indent)
  {\bfseries} % Theorem head font (bold)
  {.}         % Punctuation after theorem head
  {.5em}      % Space after theorem head
  {}          % Theorem head specification (can be left empty, meaning `normal`)

\usepackage{tikz}
\newcommand*\circled[1]{\tikz[baseline=(char.base)]{
            \node[shape=circle,draw,inner sep=1.3pt] (char) {#1};}}
\theoremstyle{boldtheorem}
\newtheorem{thm}{Theorem}%  meant for continuous numbers
\newtheorem{rem}{Remark}%
\newtheorem{cor}{Corollary}%
\newtheorem{assum}{Assumption}%
\usepackage{fancyhdr}
\usepackage{footmisc}
\usepackage[flushleft]{threeparttable}
\usepackage{xcolor}
\usepackage[thinlines]{easytable}

\begin{document}

\begin{frontmatter}
% \thispagestyle{fancy}  % Apply fancy page style only to the title page
% \fancyhf{}  % Clear any existing header/footer settings
% \fancyhead[C]{This is the header on the title page}  % Set header content for the title page

% \thispagestyle{fancy} % Use the fancy style for the first page only
% \fancyhf{} % Clear all headers and footersation.} % Centered header text



% \cortext[cor1]{Corresponding author.}
\title{Higher-Order Sinusoidal Input Describing Functions for Open-Loop and Closed-Loop Reset Control with Application to Mechatronics Systems}

\author[inst]{Xinxin Zhang}\ead{X.Zhang-15@tudelft.nl} 
\author[inst]{S. Hassan HosseinNia}\ead{S.H.HosseinNiaKani@tudelft.nl (The corresponding author). }
% \corref{cor1}
\affiliation[inst]{organization={Department of Precision and Microsystems Engineering (PME), Delft University of Technology},%Department and Organization
            addressline={Mekelweg 2}, 
            city={Delft},
            postcode={2628CD}, 
            % state={State One},
            country={The Netherlands}}
% \affiliation[inst2]{organization={Department Two},%Department and Organization
%             addressline={Address Two}, 
%             city={City Two},
%             postcode={22222}, 
%             state={State Two},
%             country={Country Two}}

\begin{abstract}
Reset control enhances the performance of high-precision mechatronics systems. This paper introduces a generalized reset feedback control structure that integrates a single reset-state reset controller, a shaping filter for tuning reset actions, and linear compensators arranged in series and parallel configurations with the reset controller. This structure offers greater tuning flexibility to optimize reset control performance. However, frequency-domain analysis for such systems remains underdeveloped. To address this gap, this study makes three key contributions: (1) developing Higher-Order Sinusoidal Input Describing Functions (HOSIDFs) for open-loop reset control systems; (2) deriving HOSIDFs for closed-loop reset control systems and establishing a connection with open-loop analysis; and (3) creating a MATLAB-based App to implement these methods, providing mechatronics engineers with a practical tool for reset control system design and analysis. The accuracy of the proposed methods is validated through simulations and experiments. Finally, the utility of the proposed methods is demonstrated through case studies that analyze and compare the performance of three controllers: a PID controller, a reset controller, and a shaped reset controller on a precision motion stage. Both analytical and experimental results demonstrate that the shaped reset controller provides higher tracking precision while reducing actuation forces, outperforming both the reset and PID controllers. These findings highlight the effectiveness of the proposed frequency-domain methods in analyzing and optimizing the performance of reset-controlled mechatronics systems.
\end{abstract}

% \let\theoddhead\relax
% \oddhead{This is the author's version which has not been fully edited and content may change prior to final public}
% \begin{figure}[h]
% 	\centering
% 	%	\missingfigure{CLCI}
% 	\includegraphics[width=\columnwidth]{figs/Graphic Abstract.pdf}
% 	\captionsetup{labelformat=empty}
% 	\caption*{Figure 0: New model for the RCS in open-loop and closed-loop.}
% 	\label{fig:graphicalabstract}
% \end{figure}

% \begin{highlights}
% \item A pulse-based model for analyzing general reset control systems (RCSs) in frequency domain is proposed.
% \item This study for the first time providing the perspective that the RCS output can be segmented into linear and non-linear elements both in open-loop and closed-loop, and further enabling the loop-shaping technique applied in the RCS analysis accurately.
% \item This study proposes an analytical method to calculate all the harmonics of RCSs  in closed-loop that can distinguish the first-order harmonic of RCSs from the DF.
% \item Sensitivity functions for reset control systems are developed.
% \end{highlights}

\end{frontmatter}
% \footnote{This is the author's version which has not been fully edited and content may change prior to final publication.}

\section{Introduction}
\label{sec: intro}
High-precision mechatronics industries require controllers capable of delivering high precision, speed, and robustness. Linear feedback controllers, particularly Proportional-Integral-Derivative (PID) controllers, are widely used in these applications due to their simplicity and effectiveness \cite{schmidt2020design}. In the PID controller, the integrator accumulates system error over time and adds it into the control signal to drive the error to zero in the long term. However, this cumulative action creates a memory effect, where even if the current error is zero or small, the integrator may still have a non-zero output due to past accumulated errors, potentially causing overshoot and stability issues. To address these challenges, the Clegg Integrator (CI) was introduced \cite{clegg1958nonlinear}, resetting the integrator state to zero whenever the error signal crosses zero. Sinusoidal-Input Describing Function (SIDF) analysis \cite{guo2009frequency} shows that the CI achieves a 51.9-degree phase lead compared to a linear integrator while maintaining the same gain characteristics. By addressing the phase-gain trade-off inherent to linear integrators, the CI enhances system performance \cite{banos2012reset}. Since then, various reset control elements have been developed, including the First-Order Reset Element (FORE), Second-Order Reset Element (SORE), Proportional-Integral (PI) + CI configurations, Hybrid Integrator-Gain systems (HIGs), and the Constant-in-Gain-Lead-in-Phase (CgLp) controller \cite{krishnan1974synthesis, horowitz1975non, banos2007definition, saikumar2019constant, van2020hybrid}. These reset controllers have been applied to improve steady-state and transient performance across diverse industries, including chemical process control, teleoperation, and mechatronics systems \cite{guo2009frequency, saikumar2019constant, banos2007design, villaverde2010reset, guo2010optimal, heertjes2016experimental, zhao2019overcoming, karbasizadeh2022continuous, banos2012reset}. This study focuses on the application of reset feedback control in high-precision mechatronics systems. 

To facilitate the practical implementation of reset control systems in the mechatronics industry, effective analysis tools are essential. Frequency response analysis is among the most commonly used techniques for this purpose in industrial applications \cite{lumkes2001control}. It evaluates a system’s steady-state response to sinusoidal inputs across varying frequencies, offering insights into phase and magnitude characteristics of linear time-invariant (LTI) systems. Frequency response analysis covers both open-loop and closed-loop analysis. By leveraging the connection between the open-loop and closed-loop analysis through loop-shaping techniques \cite{grassi2001integrated}, control engineers can design controllers in the open loop, ensuring that the system meets specified closed-loop performance requirements, such as reducing steady-state errors and improving transient response \cite{richter2011advanced}. Additionally, frequency response analysis allows engineers to predict closed-loop behavior without requiring precise parametric models of the plant. This characteristic is particularly beneficial when obtaining an accurate plant model is impractical.
% However, frequency response analysis for reset control systems is challenging due to the higher-order harmonics. Additionally, higher-order harmonics in the output signal, when fed back into the system, can iteratively generate additional harmonics that affect the outputs. Thus, closed-loop Additionally, frequency response analysis allows engineers to predict closed-loop behavior without requiring precise parametric models of the plant. This characteristic is particularly beneficial when obtaining an accurate plant model is challenging or impractical.

% Recently, literature has made strides in addressing these challenges. 
For frequency response analysis of open-loop reset feedback control systems, Higher-Order Sinusoidal Input Describing Function (HOSIDF) methods, as detailed in \cite{nuij2006higher, saikumar2021loop, karbasizadeh2022band, van2024higher}, are employed. These HOSIDF analysis methods align with the SIDF analysis method when the high-order (beyond the first-order) harmonics are negligible \cite{guo2009frequency}. However, the accuracy of existing HOSIDF methods for open-loop reset control systems is constrained to configurations where the input signal and the reset-triggered signal for the reset controller are identical. Currently, no accurate open-loop HOSIDF analysis method is available for the more generalized reset control systems, as presented in this study.
% In systems where these signals differ, the analysis presented in \cite{karbasizadeh2022band} demonstrates inaccuracies in the phase of high-order harmonics.
% In systems where these signals differ, the analysis presented in \cite{karbasizadeh2022band} demonstrates inaccuracies, due to incorrect phase calculations for high-order harmonics. This limitation is further detailed in the Problem Statement Section \ref{subsec: open-loop problem}.

For closed-loop reset feedback control systems, frequency response analysis is particularly challenging because high-order harmonics can generate additional harmonics through the feedback loop, complicating the system dynamics and violating the superposition. Research in \cite{saikumar2021loop} introduced the HOSIDF method for such systems, establishing a connection between open-loop and closed-loop analyses, but it neglected the effects of reset actions on high-order harmonics within the feedback loop, resulting in inaccuracies. To address this, our recent work \cite{ZHANG2024106063} proposed an improved HOSIDF method that corrects these inaccuracies. However, the approach remains limited to specific reset control structures.

Motivated by the limitations in open-loop and closed-loop frequency response analysis for reset feedback control systems, this study makes the following contributions:
\begin{itemize}
    \item First, this study introduces a generalized reset control structure that incorporates a single reset-state reset controller, along with a shaping filter to tune reset actions, and linear compensators positioned in series before and after, as well as in parallel with, the reset controller. This structure broadens the tuning capabilities of reset control. Then, building on prior work in limited reset configurations \cite{ZHANG2024106063}, two frequency response analysis tools are developed for this structure: (1) open-loop Higher-Order Sinusoidal Input Describing Functions (HOSIDFs) and (2) closed-loop HOSIDFs for systems under the two-reset conditions \cite{Xinxin_multiple_reset}. Furthermore, a frequency-domain link is established between open-loop and closed-loop analyses using HOSIDFs. The effectiveness of these methods is validated through simulations and experiments on a precision motion stage.
    \item Then, the open-loop and closed-loop HOSIDFs for reset control systems are integrated into a MATLAB App, offering control engineers a practical, user-friendly tool for reset control systems analysis and design. 
    % Moreover, detailed guidance to demonstrate the App’s functionality is provided.
    \item Finally, case studies are presented to demonstrate the performance capabilities of the generalized reset control structure and the effectiveness of the proposed frequency response analysis methods. Using the HOSIDFs methods, the performance of three controllers—PID, CgLp \cite{saikumar2019constant}, and shaped CgLp—is analyzed. Frequency-domain analysis results reveal that both the CgLp and shaped CgLp controllers provide phase lead compared to the PID controller. Furthermore, the shaped CgLp controller effectively reduces high-order harmonics while retaining the benefits of the first-order harmonic compared to the CgLp controller. These frequency-domain enhancements enable the shaped CgLp controller to achieve the lowest steady-state error and actuation force among the three controllers, validated through precision motion stage experiments.
\end{itemize}

The remainder of the paper is organized as follows. Section \ref{sec:preliminaries} introduces the generalized reset feedback control structure and the experimental setup used in this study. Section \ref{sec: Toolbox 1} presents the HOSIDFs for open-loop reset control systems, followed by Section \ref{sec: Toolbox 3}, which details the HOSIDFs for closed-loop reset control systems. Section \ref{sec:matlab app} consolidates the methods from Sections \ref{sec: Toolbox 1} and \ref{sec: Toolbox 3} into a MATLAB App. Section \ref{sec: Application of Theorems} demonstrates the application of the proposed methods in analyzing the performance of reset controllers on a precision motion stage. Finally, Section \ref{sec: conclusion} presents concluding remarks and outlines future research directions.

\section{Preliminaries}
\label{sec:preliminaries}
This section begins by defining the generalized reset control system. Following this, the stability and convergence conditions for the reset control system are outlined. Finally, the experimental precision motion stage used in this work is introduced.
% Next, sinusoidal-input analysis for the reset controller is provided. 
\subsection{A Generalized Reset Feedback Control System}
This study focuses on the frequency-domain analysis of a generalized reset feedback control system, whose block diagram is defined in Fig. \ref{fig1:RC system}. 
% This structure is designed to extend the capabilities of reset controllers beyond those of traditional reset control systems.
\begin{figure}[htp]
\centering
	%	\missingfigure{RCsystem} 
         % \captionsetup{singlelinecheck = false, format= hang, justification=justified, font=footnotesize, labelsep=space}
           \centering
	\includegraphics[width=0.75\columnwidth]{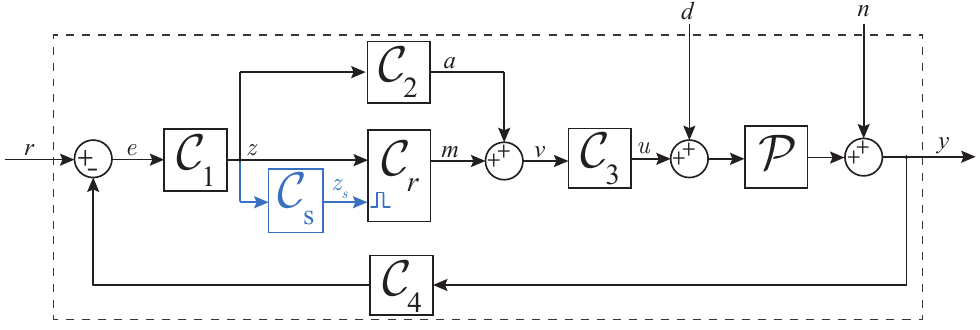}
	\caption{Block diagram of a generalized reset control system, with the resetting action denoted by blue lines.}
	\label{fig1:RC system}
\end{figure}

In this configuration, \(r\), \(d\), \(n\), \(e\), \(u\), and \(y\) represent the reference input, disturbance, noise, error, control input, and system output signals, respectively. The block \(\mathcal{C}_r\) represents the reset controller, while the LTI shaping filter \(\mathcal{C}_s\) generates the reset-triggered signal \(z_s\) to trigger the reset actions. Systems \(\mathcal{C}_1\), \(\mathcal{C}_2\), and \(\mathcal{C}_3\) are LTI controllers integrated into the feed-through loop, while the LTI controller \(\mathcal{C}_4\) is placed within the feedback loop. The plant is denoted by \(\mathcal{P}\). 

The reset controller \(\mathcal{C}_r\) is a hybrid system that combines a linear controller with a reset mechanism \cite{banos2012reset, guo2019stability}. The state-space representation of \(\mathcal{C}_r\), with state \(x_r(t) \in \mathbb{R}^{n_c\times 1}\), input \(z(t)\), and output \(m(t)\), is given by:
\begin{equation} 
	\label{eq: State-Space eq of RC} 
	\mathcal{C}_r = 
 \begin{cases}
        \dot{x}_r(t) = A_Rx_r(t) + B_Rz(t), &  t \notin J, \\
		x_r(t^+) = A_\rho x_r(t), &  t \in J, \\
		m(t) = C_Rx_r(t) + D_Rz(t),
  \end{cases}
\end{equation} 
where matrices \(A_R \in \mathbb{R}^{n_c \times n_c}\), \(B_R \in \mathbb{R}^{n_c \times 1}\), \(C_R \in \mathbb{R}^{1 \times n_c}\), and \(D_R \in \mathbb{R}^{1 \times 1}\) define the flow dynamics of the reset controller \(\mathcal{C}_r\), referred to as the Base-Linear Controller (BLC) \(\mathcal{C}_{l}\), and are represented by:
\begin{equation}
	\label{eq: RL}
	C_{l}(\omega) = C_R(j\omega I - A_R)^{-1}B_R + D_R,
\end{equation}
where $\omega\in\mathbb{R}^+$ represents the angular frequency in the frequency domain. Replacing \(\mathcal{C}_r\) with \(\mathcal{C}_l\) \eqref{eq: RL}, the system in Fig. \ref{fig1:RC system} is termed the Base-Linear System (BLS). 

The reset controller \(\mathcal{C}_r\) in \eqref{eq: State-Space eq of RC} employs the ``zero-crossing law" as its reset mechanism \cite{banos2012reset}, where the state \(x_r(t)\) is reset to \(x_r(t^+)\) whenever the reset trigger signal \(z_s(t)\) crosses zero. Therefore, the set of reset instants is defined as \(J = \{t_i \mid z_s(t_i) = 0, i \in \mathbb{N}\}\). At each reset instant \(t_i \in J\), the jump dynamics of \(\mathcal{C}_r\) are determined by the reset matrix \(A_\rho\), given by
\begin{equation}
\label{eq: A_rho}
    A_\rho = \begin{bmatrix}
		\gamma & \\
		& I_{n_c - 1}
	\end{bmatrix}, \quad \gamma \in (-1, 1].
\end{equation}
Equation \eqref{eq: A_rho} defines reset controllers with a single reset state. Common examples of such reset elements include the CI, the FORE, and the Second-Order Single State Reset Element (SOSRE) \cite{karbasizadeh2021fractional}. When $\gamma=1$ and thus \(A_\rho = I_{n_c}\) in \eqref{eq: A_rho}, the reset controller \(\mathcal{C}_r\) is identical to \(\mathcal{C}_l\) in \eqref{eq: RL}.

\subsection{Stability and Convergence Conditions for Reset Control Systems}
This paper works on the development of frequency response analysis methods for reset feedback control systems. Although stability and convergence conditions are not the primary focus of this paper, they are needed for frequency response analysis \cite{pavlov2006uniform, pavlov2007frequency}.

Following established literature, we introduce Assumptions \ref{assum: open-loop stability} and \ref{assum: closed-loop stability} to ensure the stability and convergence conditions for open-loop and closed-loop reset control systems, respectively.

The literature \cite{guo2009frequency} demonstrates that the reset controller defined in \eqref{eq: State-Space eq of RC}, when subjected to an input \( z(t) = |Z|\sin(\omega t + \angle Z) \), where \( |Z| \) and \( \angle Z \) denote the magnitude and phase of the signal \( z(t) \) respectively, exhibits a globally asymptotically stable \( 2\pi/\omega \)-periodic solution and converges globally if and only if: 
\begin{equation}
\label{eq:open-loop stability}
    |\lambda (A_\rho e^{A_R\delta})|<1,\ \forall \delta \in \mathbb{R}^+,
\end{equation} 
where $\lambda(\cdot)$ represents the eigenvalues of the matrix.

To ensure the the HOSIDF analysis for open-loop reset control systems, the following assumption is introduced:
\begin{assum}
\label{assum: open-loop stability}
The reset controller \(\mathcal{C}_r\) \eqref{eq: State-Space eq of RC}, with an input \( z(t) = |Z|\sin(\omega t + \angle Z) \), is assumed to satisfy the condition in \eqref{eq:open-loop stability}. Additionally, LTI systems $\mathcal{C}_1$, $\mathcal{C}_2$, $\mathcal{C}_3$, $\mathcal{C}_4$, and $\mathcal{C}_s$ are Hurwitz.
\end{assum}

% Defined $\mathscr{L}_2$ as the space of square-integrable functions on the time axis $t \geq 0$, with the inner product and norm given by
% \begin{equation}
%   <w,m> \ :=\int_0^\infty z(t)m(t)dt, \text{ and } ||w|| = \sqrt{<w,w>}.  
% \end{equation}
% Let \(\mathscr{L}_2\) denote the space of all measurable functions \( f(\cdot): \mathbb{R}_+ \to \mathbb{R} \) for which \( \int_0^\infty |f(t)|^2 \, dt < \infty \). The \(\mathscr{L}_2\)-norm, \( ||\cdot||: \mathscr{L}_2 \to \mathbb{R}_+ \), is defined as \( ||f|| = \sqrt{\int_0^\infty |f(t)|^2 \, dt} \). According to \cite{carrasco2008stability}, a closed-loop reset control system is \(\mathscr{L}_2\)-stable if, for every \(\mathscr{L}_2\)-stable input signal, the output \( y \) also belongs to \(\mathscr{L}_2\). The \(\mathscr{L}_2\)-stability of the open-loop reset controller can be achieved by meeting the known \( H_{\beta} \) condition \cite{beker2000forced, beker2004fundamental}.

To facilitate the HOSIDF analysis of the closed-loop reset control system, the following assumption is introduced, ensuring the uniform exponential convergence of the closed-loop reset control system as established in \cite{dastjerdi2022closed}:
\begin{assum}
\label{assum: closed-loop stability}
The initial condition of the reset controller \(\mathcal{C}_r\) \eqref{eq: State-Space eq of RC} is zero, there are infinitely many reset instants \( t_i \) with \( \lim\limits_{i \to \infty} t_i = \infty \), the input signals are Bohl functions \cite{barabanov2001bohl}, and the \( H_\beta \) condition detailed in \cite{beker2000forced, beker2004fundamental} is satisfied.
\end{assum}

Assumption \ref{assum: closed-loop stability} can be achieved through appropriate system design \cite{banos2012reset, saikumar2021loop}. When this assumption is satisfied, the closed-loop reset control system in Fig. \ref{fig1:RC system}, driven by a sinusoidal input with a frequency of $\omega$, attains a periodic steady-state response. This steady-state behavior can be described by the expression \( x(t) = \mathcal{S}(\sin(\omega t), \cos(\omega t), \omega) \), where \( \mathcal{S}: \mathbb{R}^3 \to \mathbb{R}^{n_{cl}} \), represents a function of the input signal and its frequency \cite{dastjerdi2022closed}, with \( n_{cl} \) denoting the number of states in closed-loop reset control systems.
\subsection{Precision Positioning Setup}
\label{sec:spider}
This study introduces frequency response analysis methods for generalized reset control systems in Fig. \ref{fig1:RC system}. Accurate frequency response analysis is essential for the effective design and analysis of reset control systems in precision motion control. For example, prior work in \cite{karbasizadeh2022continuous} developed a continuous CgLp element to suppress oscillations in precision motion systems, but this relied on parameter-specific and computationally intensive numerical methods due to the lack of closed-loop frequency response analysis techniques. The proposed HOSIDFs for open-loop and closed-loop reset control systems, along with their connection, provide magnitude and phase information across the entire frequency range, facilitating systematic optimization of reset-controlled mechatronics systems.
% Reset controllers have demonstrated significant effectiveness in enhancing mechatronics performance metrics, including tracking precision, disturbance rejection, and noise suppression \cite{banos2012reset}. 

The experimental setup is a three-Degree-of-Freedom (3-DoF) precision positioning stage, depicted in Fig. \ref{fig1:Spyder}. The stage consists of three masses, \(M_1\), \(M_2\), and \(M_3\), which are connected to a central base mass \(M_c\) via dual leaf flexures. Each mass is actuated by its respective voice coil actuator, labeled \(A_1\), \(A_2\), and \(A_3\). Position feedback for the masses is obtained using Mercury M2000 linear encoders (denoted as ``Enc"), which offer a resolution of 100 nm and are sampled at a frequency of 10 kHz. Control systems are implemented on an NI CompactRIO platform, equipped with a linear current source power amplifier.
\begin{figure}[!t]
	%	\missingfigure{RCsystem}
    \centering
        % \captionsetup{singlelinecheck = false, format= hang, justification=justified, font=footnotesize, labelsep=space}
	\centerline{\includegraphics[width=0.75\columnwidth]{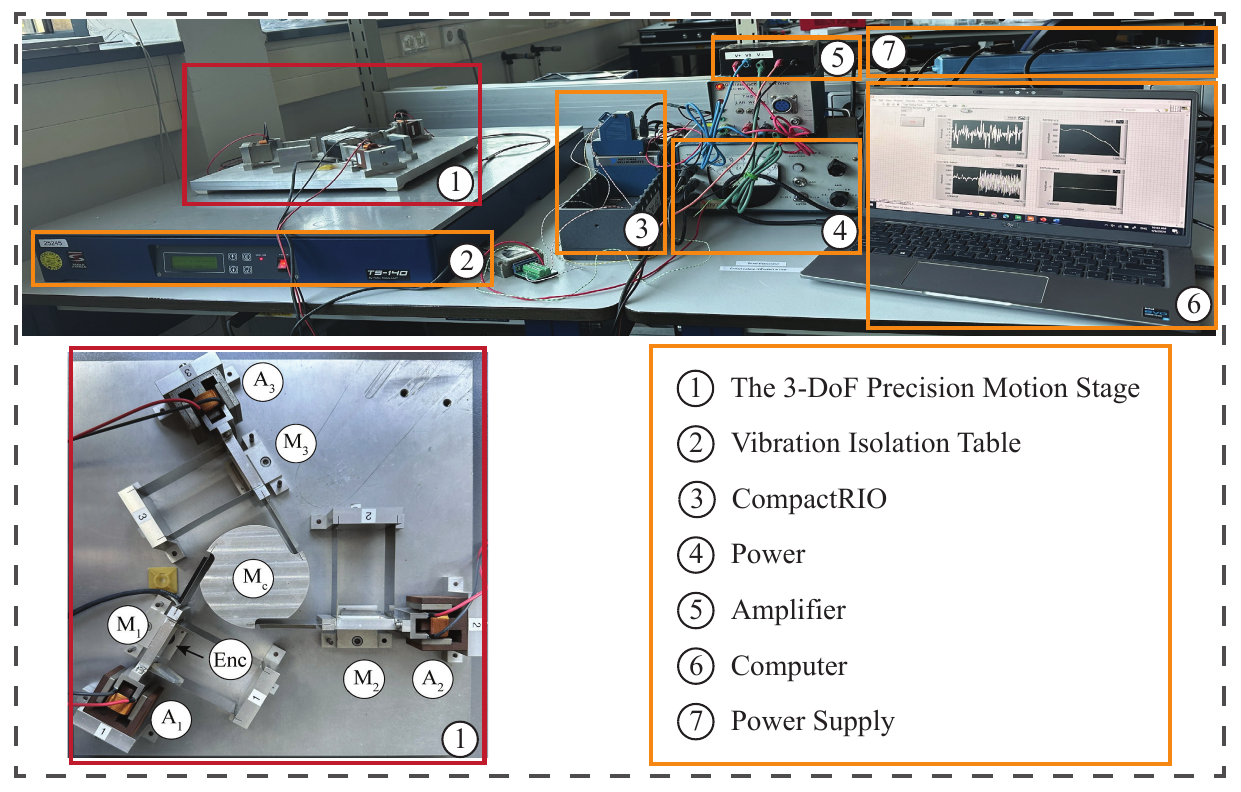}}
	\caption{The planar precision positioning stage.}
	\label{fig1:Spyder}
\end{figure}

In this study, only actuator \(A_1\) is employed to control the position of mass \(M_1\). Figure \ref{fig1:Spyder_frf} illustrates the measured Frequency Response Function (FRF) of the system, which closely resembles that of a linear LTI collocated double mass-spring-damper system, albeit with additional high-frequency parasitic dynamics. Utilizing the system identification toolbox in MATLAB, the system's main dynamics are modeled by the following transfer function \(\mathcal{P}(s)\), described as:
\begin{equation}
\label{eq:P(s)}
    \mathcal{P}(s) = \frac{6.615\times10^5}{83.57s^2+279.4s+5.837\times10^5}.
\end{equation}
This simplified model effectively captures the essential mass-spring-damper dynamics of the system.
\begin{figure}[!t]
	%	\missingfigure{RCsystem}
    \centering
        % \captionsetup{singlelinecheck = false, format= hang, justification=justified, font=footnotesize, labelsep=space}
	\centerline{\includegraphics[width=0.75\columnwidth]{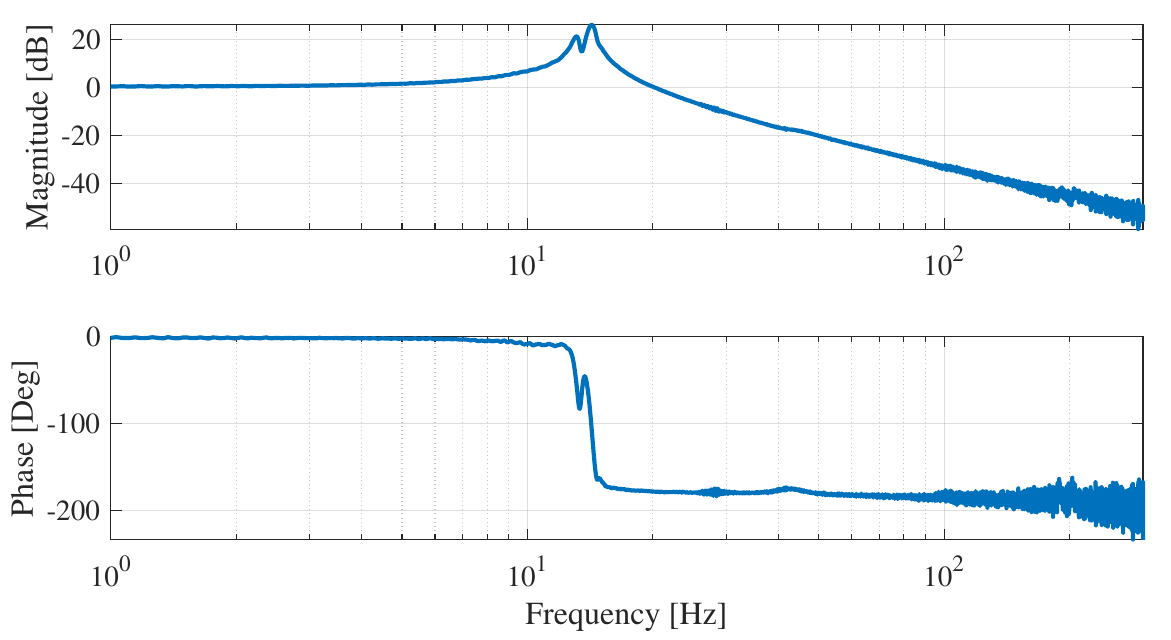}}
	\caption{The FRF data from actuator $A_1$ to attached mass $M_1$.}
	\label{fig1:Spyder_frf}
\end{figure}

% \input{2_Problem_Statement}
% \section{Toolbox \rom{1}: The Analysis for Open-Loop reset control sytems}
\section{Main Result 1: Frequency Response Analysis Method and Validation for Open-Loop Reset Control Systems}
\label{sec: Toolbox 1}
Figure \ref{fig: open_loop_rcs} depicts the block diagram of the open-loop reset control system.
 % In this configuration, the reset controller \(\mathcal{C}_r\) \eqref{eq: State-Space eq of RC} receives an input signal \(z_o(t) = |Z_o|\sin(\omega t + \angle Z_o)\) and generates an output signal \(m_o(t)\), assuming that Assumption \ref{assum: open-loop stability} is satisfied. The reset-triggered signal \(z_s(t)\) is defined as \(z_s(t) = |Z_o \mathcal{C}_s(\omega)|\sin(\omega t + \angle Z_o + \angle \mathcal{C}_s(\omega)\), which is obtained by filtering \(z_o(t)\) through the transfer function \(\mathcal{C}_s\).
\begin{figure}[H]
	%	\missingfigure{RCsystem} 
    \centering
         % \captionsetup{singlelinecheck = false, format= hang, justification=justified, font=footnotesize, labelsep=space}
           \centering
	\includegraphics[width=0.75\columnwidth]{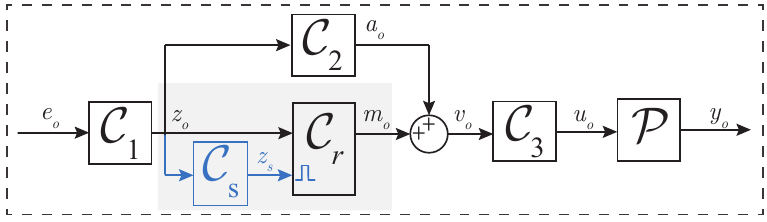}
	\caption{Block diagram of the open-loop reset control system.}
	\label{fig: open_loop_rcs}
\end{figure}

This section extends our previous work \cite{ZHANG2024106063} on frequency response analysis for open-loop reset control systems where \(\mathcal{C}_1 = \mathcal{C}_s = \mathcal{C}_3 = \mathcal{C}_4 = 1\) and \(\mathcal{C}_2 = 0\) to generalized reset control systems in Fig. \ref{fig: open_loop_rcs}. 
\subsection{HOSIDFs for the Open-Loop Reset Control Systems}

 The HOSIDFs analysis is a technique used for analyzing the frequency response of nonlinear systems \cite{nuij2006higher}. Theorem \ref{thm: open-loop HOSIDF} provides the HOSIDFs for the open-loop reset control system in Fig. \ref{fig: open_loop_rcs}.
\begin{thm}
\label{thm: open-loop HOSIDF}
Consider an open-loop reset control system as shown in Fig. \ref{fig: open_loop_rcs}, with an input signal \(e_o(t) = |E| \sin(\omega t + \angle E)\), resulting in the output signal \(y_o(t)\) under Assumption \ref{assum: open-loop stability}. Using the ``Virtual Harmonic Generator" \cite{nuij2006higher}, the input signal \(e_o(t)\) generates harmonics expressed as \(e_o^n(t) = |E| \sin(n\omega t + n \angle E)\), with the corresponding Fourier transform denoted as \(E_o^n(\omega)\). The signals \(z_o(t)\), \(m_o(t)\), and \(y_o(t)\) consist of \(n\) harmonics, represented as \(z_o^n(t)\), \(m_o^n(t)\), and \(y_o^n(t)\), with Fourier transforms \(Z_o^n(\omega)\), \(M_o^n(\omega)\), and \(Y_o^n(\omega)\), respectively. The Higher-Order Sinusoidal Input Describing Functions (HOSIDFs) of the reset controller \(\mathcal{C}_r\) are given by
% describes the transfer function from \(Z_o^n(\omega)\) to \(M_o^n(\omega)\), given by:
\begin{equation}
\label{eq: Cr_hn_thm} 
\begin{aligned}
\mathcal{C}_r^n(\omega)&= \frac{M_o^n(\omega)}{Z_o^n(\omega)}= \begin{cases}
        \mathcal{C}_{l}(\omega) + \mathcal{C}_{\rho}^1(\omega), & \text{for }n=1, \\
	\mathcal{C}_{\rho}^n(\omega), & \text{for odd }n>1,\\
		0,&\text{for even }n \geqslant 2,
  \end{cases}    
\end{aligned}
\end{equation}
and the HOSIDFs of the open-loop reset control system are given by
% The HOSIDF analysis of the open-loop reset control system, which describes the transfer function from \(E_o^n(\omega)\) to \(Y_o^n(\omega)\), is given as follows:
\begin{equation}
\label{eq: H_hn} 
\begin{aligned}
&\mathcal{L}_n(\omega) =\frac{Y_o^n(\omega)}{E_o^n(\omega)}=
% \\&
 % \resizebox{.5\columnwidth}{!}{$
 \begin{cases}
\mathcal{C}_{1}(\omega)[\mathcal{C}_{l}(\omega) + \mathcal{C}_{\rho}^1(\omega)+ \mathcal{C}_2(\omega)]\mathcal{C}_{3}(\omega)\mathcal{P}(\omega), & \text{for }n=1, \\
\mathcal{C}_{1}(\omega)e^{j(n-1)\angle\mathcal{C}_{1}(\omega)}\mathcal{C}_{\rho}^n(\omega)\mathcal{C}_{3}(n\omega)\mathcal{P}(n\omega), & \text{for odd }n>1,\\
0,&\text{for even }n\geqslant 2,
\end{cases}
% $}
\end{aligned}
\end{equation}
where 
\begin{equation}
\label{eq: Delta_l, Delta_x, Delta_c, Delta_q, C_rho_n}
\begin{aligned}
\Delta_l(\omega) &= (j\omega I-A_R)^{-1}B_R,\\
\Delta_x(n\omega) &= C_R(jn\omega I-A_R)^{-1}jn\omega I,\\
\Delta_c(\omega) &= |\Delta_l(\omega)|\sin(\angle\Delta_l(\omega) -\angle \mathcal{C}_s(\omega)),\\
\mathcal{C}_{\rho}^n(\omega) &=  2\Delta_x(n\omega)\Delta_q(\omega) e^{jn\angle \mathcal{C}_s(\omega)}/(n\pi),\\
\Delta_q(\omega) &= 
% \resizebox{.39\columnwidth}{!}{$
(I+e^{A_R\pi/\omega})(A_\rho e^{A_R\pi/\omega}+I)^{-1}(A_\rho-I)\Delta_c(\omega). 
% \mathcal{L}_{\rho}(n\omega) &= \mathcal{C}_{\rho}^n(\omega)\mathcal{C}_3(n\omega)\mathcal{P}(n\omega)\mathcal{C}_4(n\omega)\mathcal{C}_1(n\omega),\\
% \mathcal{L}_{l}(n\omega) &= {[\mathcal{C}_{l}(n\omega)+\mathcal{C}_{2}(n\omega)]\mathcal{C}_{3}(n\omega)\mathcal{P}(n\omega)\mathcal{C}_{4}(n\omega)}\mathcal{C}_{1}(n\omega).
% $}. 
\end{aligned}
\end{equation}
\end{thm}
\begin{proof}
The proof is provided in \ref{Proof: thmCv}.
\end{proof}
% \begin{rem}
% When the high-order harmonics \(n>1\) are negligible, the HOSIDF analysis in Theorem \ref{thm: closed-loop HOSIDF} is identical to the SIDF analysis in \cite{guo2009frequency}.
%  \end{rem}

Based on Theorem \ref{thm: open-loop HOSIDF} and its proof in \ref{Proof: thmCv}, Fig. \ref{fig1:OL_Block_Diagram} illustrates the block diagram of the open-loop reset control system for HOSIDF analysis. Following this, Remark \ref{rem: yo(t)} provides the calculation for the output \( y_o(t) \) of the sinusoidal-input open-loop reset control system.
\begin{figure}[!t]
    \centering
	%	\missingfigure{RCsystem}
         % \captionsetup{singlelinecheck = false, format= hang, justification=justified, font=footnotesize, labelsep=space}
	\centerline{\includegraphics[width=0.75\columnwidth]{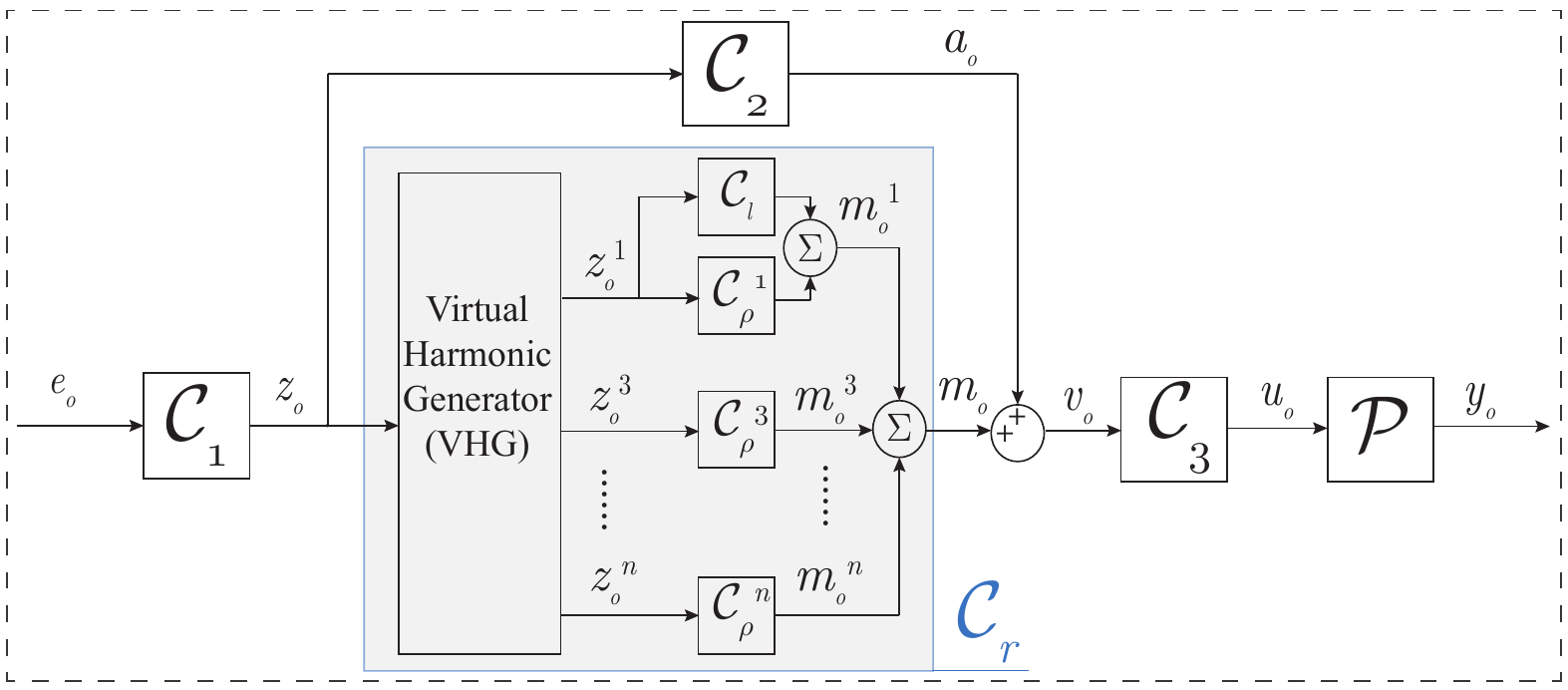}}
	\caption{Block diagram of the open-loop reset control system for HOSIDF analysis.}
	\label{fig1:OL_Block_Diagram}
\end{figure} 
% \begin{cor}
% \label{cor: m = m_l + m_rho}
% In a SISO open-loop reset control system as shown in Fig. \ref{fig1:RC system} with a input signal of $e_o(t) = |E|\sin(\omega t + \angle E)$, the output signal $m_o(t)$ of $\mathcal{C}_r$ is composed of linear element denoted as $m_l(t)$ and nonlinear elements denoted as $m_\rho(t)$, given by
% \begin{equation}
% \label{eq: mo,ml,mrho}
% \begin{aligned}    
%     m_o(t) &= m_l(t) + m_\rho(t),\\
%     m_l(t) &= |\mathcal{C}_l(\omega)|Z_o^1(t)e^{j\angle \mathcal{C}_l(\omega)},\\
%     m_\rho(t) &= \sum\nolimits_{n=1}^{\infty}|\mathcal{C}_\rho^n(\omega)|Z_o^n(t)e^{j\angle \mathcal{C}_\rho^n(\omega)}.
% \end{aligned}
% \end{equation}
% \end{cor}
\begin{rem}
    \label{rem: yo(t)}
Consider an open-loop reset control system with the input signal \(e_o(t) = |E|\sin(\omega t+ \angle E)\), under Assumption \ref{assum: open-loop stability}. The steady-state output signal \(y_o(t)\) is given by
\begin{equation}
\label{cor_eq_y_o(t)}
\begin{aligned}
y_o(t) = \sum\nolimits_{n=1}^{\infty} y_o^n(t) = \sum\nolimits_{n=1}^{\infty} |E\mathcal{L}_n(\omega)|&\sin(n\omega t+ n\angle E + \angle\mathcal{L}_n(\omega)),
%  \\ &
n=2k+1 (k\in\mathbb{N}).    
\end{aligned}
\end{equation}
\end{rem} 
% \subsection{Toolbox \rom{1} Results: Analysis for Open-Loop reset control systems}
% \subsection{Results: The Validation of Open-Loop Analysis}
\subsection{Validation of the Open-Loop HOSIDFs}

This section uses an illustrative system to verify the accuracy of the open-loop HOSIDF, \(\mathcal{L}_n(\omega)\), derived in \eqref{eq: H_hn}. The illustrative system is based on the structure shown in Fig. \ref{fig: open_loop_rcs}, with the following design parameters: the reset controller \(\mathcal{C}_r\) is based on a BLC \(\mathcal{C}_l = 30\pi/s\) with a reset value \(\gamma = 0\), \(\mathcal{C}_1 = (s/(150\pi))/(s/(3000\pi)+1)\), \(\mathcal{C}_2 = \mathcal{C}_4 = 1\), \(\mathcal{C}_s = 1/(s/5+1)\), and \(\mathcal{C}_3 = 1/(s/(150\pi)+1)\). The plant $\mathcal{P}$ is given in \eqref{eq:P(s)}.

The input to the system is a sinusoidal signal \(e_o(t) = \sin(8\pi t)\). Figure \ref{fig:ex4_yo}(a) illustrates the output signal \(y_o(t) = \sum_{n=1}^{399} y_o^n(t)\) along with its first five harmonic components \(y_o^n(t)\) (\(n=1, 3, 5, 7, 9\)), computed using Theorem \ref{thm: open-loop HOSIDF} and Remark \ref{rem: yo(t)}. Moreover, Figure \ref{fig:ex4_yo}(b) compares the output signal \(y_o(t)\) obtained from simulation with the prediction generated by the HOSIDFs analysis method. The close agreement between the simulated and predicted results demonstrates the accuracy of the proposed HOSIDF analysis method for predicting the behavior of open-loop reset control systems.
\begin{figure}[!t]
    \centering
    % \captionsetup{singlelinecheck = false, format= hang, justification=justified, font=footnotesize, labelsep=space}
    \includegraphics[width=0.75\columnwidth]{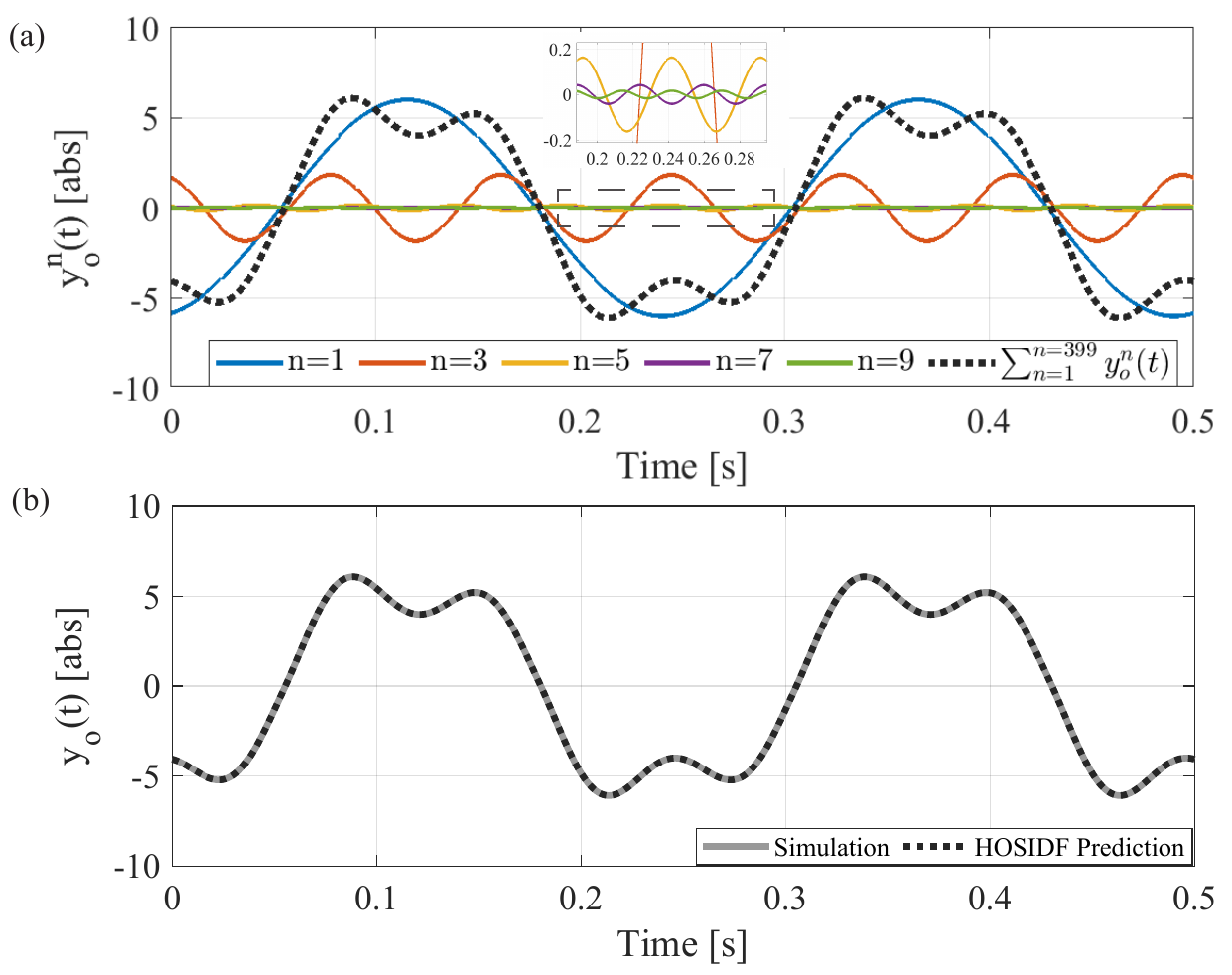}
    \caption{(a) The output signal \(y_o(t) = \sum_{n=1}^{399} y_o^n(t)\) and its first five harmonics \(y_o^n(t)\) (for \(n=1,3,5,7,9\)) for the illustrative open-loop reset control system under a sinusoidal input \( e_o(t) = \sin(8\pi t) \), obtained based on Theorem \ref{thm: open-loop HOSIDF}. (b) Simulated, previous prediction \cite{karbasizadeh2022band}, and Theorem \ref{thm: open-loop HOSIDF}-predicted output signal $y_o(t)$.}
	\label{fig:ex4_yo}
\end{figure}

The accuracy of the HOSIDFs analysis method in Theorem \ref{thm: open-loop HOSIDF} depends on the number of harmonics denoted as $N_h$ included in the analysis. Define the prediction error as the difference between the prediction provided by Theorem \ref{thm: open-loop HOSIDF} and the simulation results. Figure \ref{fig:ex4_pe_nh_final} illustrates the relationship between the prediction error and the number of harmonics $N_h$. The results demonstrate that incorporating a higher number of harmonics in the calculations enhances prediction accuracy. Given that the true nonlinear output signal \( y_o(t) \) of the reset control system contains an infinite number of harmonics, ideally, as the number of harmonics approaches infinity, the prediction error converges to zero.
\begin{figure}[!t]
    \centering
    % \captionsetup{singlelinecheck = false, format= hang, justification=justified, font=footnotesize, labelsep=space}
    \includegraphics[width=0.7\columnwidth]{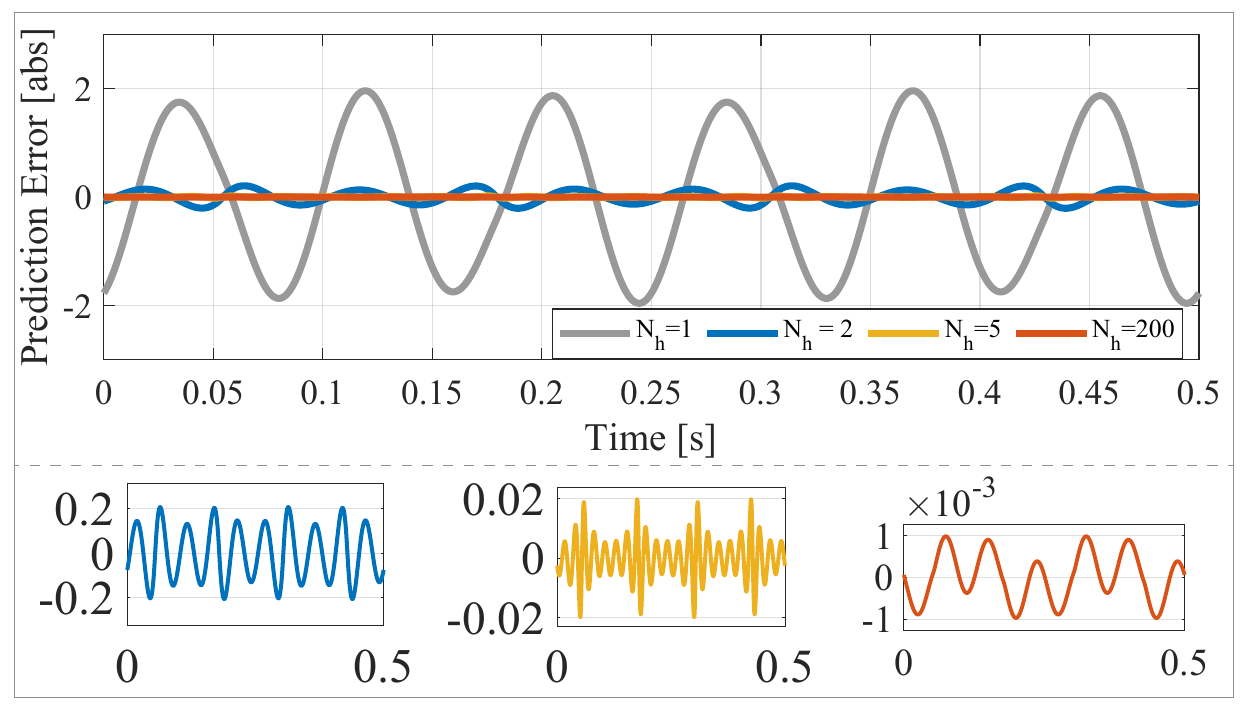}
    \caption{The relationship between the prediction error and the number of harmonics \( N_h \) considered in the calculation, with values \( N_h = 1 \), \( 2 \), \( 10 \), and \( 200 \).}
	\label{fig:ex4_pe_nh_final}
\end{figure}

After validating the accuracy of the open-loop analysis method, Theorem \ref{thm: open-loop HOSIDF} is used to perform a frequency-domain analysis of the open-loop reset control systems depicted in Fig. \ref{fig: open_loop_rcs}. Figure \ref{fig:ex4_Lo} shows the Bode plot of the open-loop HOSIDFs \(\mathcal{L}_n(\omega)\) for the illustrative open-loop reset control system. These HOSIDFs provide critical magnitude and phase information for each harmonic, which is essential for the effective design and optimization of the system. 
\begin{figure}[!t]
    \centering
    % \captionsetup{singlelinecheck = false, format= hang, justification=justified, font=footnotesize, labelsep=space}
    \includegraphics[width=0.75\columnwidth]{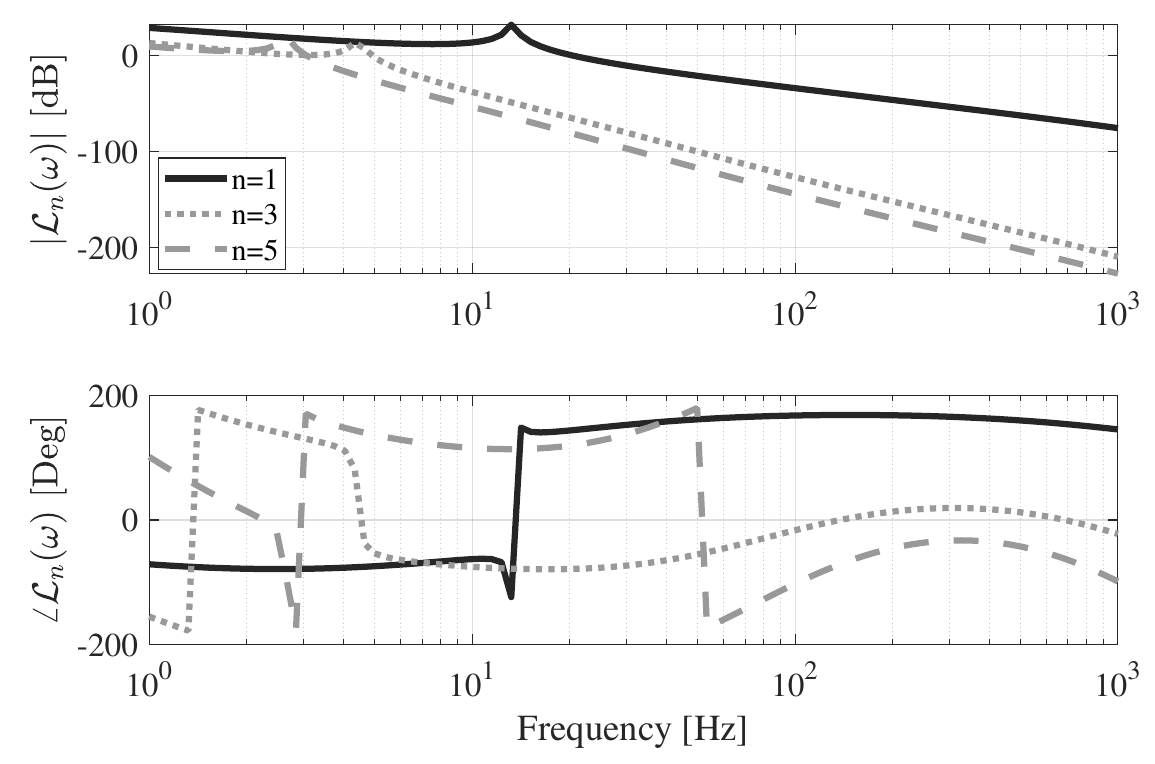}
    \caption{The HOSIDF $\mathcal{L}_n(\omega)$ of the open-loop reset control system with the first ($n=1$), third ($n=3$), and fifth ($n=5$) order harmonics.}
	\label{fig:ex4_Lo}
\end{figure}

To summarize this section, Theorem \ref{thm: open-loop HOSIDF} present accurate HOSIDF analysis for reset controllers and open-loop reset control systems. More importantly, these methods analytically decompose the HOSIDFs of the reset controller $\mathcal{C}_r$ into its base-linear transfer function $\mathcal{C}_l(\omega)$ and nonlinear components $\mathcal{C}_\rho^n(\omega)$ in \eqref{eq: Cr_hn_thm}. This decomposition serves as the foundation for the development of the closed-loop HOSIDFs, which will be elaborated in Section \ref{sec: Toolbox 3}.

% The first-order harmonic obtained from Theorem \ref{thm: open-loop HOSIDF} matches the result from equation \eqref{eq: Hn}; however, the phases of their higher-order harmonics are different. This discrepancy arises because the previous method omits the phase shift caused by $\mathcal{C}_1$ in Fig. \ref{fig: open_loop_rcs}.

% Moreover, the main contribution of Lemma \ref{lem: Cr_n} and Theorem \ref{thm: open-loop HOSIDF} is the separation of the HOSIDF of the reset controller \(\mathcal{C}_r\) into its linear element \(\mathcal{C}_l\) and its nonlinear elements \(\mathcal{C}_\rho^n\). This separation facilitates the development of the closed-loop HOSIDF analysis in Section \ref{sec: Toolbox 3}.

% \section{Toolbox Part \rom{3}: The Frequency Response Analysis for the Closed-loop reset control systems}
\section{Main Result 2: Frequency Response Analysis Method and Validation for Closed-loop Reset Control Systems}
\label{sec: Toolbox 3}
\subsection{HOSIDFs for the Closed-Loop Reset Control Systems}
This section extends the closed-loop HOSIDF method, where \(\mathcal{C}_1 = \mathcal{C}_s = \mathcal{C}_3 = \mathcal{C}_4 = 1\) and \(\mathcal{C}_2 = 0\) as presented in \cite{ZHANG2024106063}, to the generalized reset control systems with two reset actions per steady-state cycle.

In a closed-loop reset control system, as depicted in Fig. \ref{fig1:RC system}, and under the conditions outlined in Assumption \ref{assum: closed-loop stability}, when the system is subjected to a single sinusoidal input signal of frequency \(\omega\), the resulting signals \(e(t)\) (error), \(z(t)\) (input to the reset controller), \(z_s(t)\) (reset-triggered signal), \(u(t)\) (control input), and \(y(t)\) (output) become periodic and share the same fundamental frequency as the input signal \cite{dastjerdi2022closed, pavlov2005convergent}, expressed as:
\begin{equation}
\label{eq: e,y,u}
    \begin{aligned}
        e(t) &= \sum\nolimits_{n=1}^{\infty}e^n(t) = \sum\nolimits_{n=1}^{\infty} |E^n|\sin(n\omega t+\angle E^n),\\
        z(t) &= \sum\nolimits_{n=1}^{\infty}z^n(t) = \sum\nolimits_{n=1}^{\infty} |Z^n|\sin(n\omega t+\angle Z^n),\\
        z_s(t) &= \sum\nolimits_{n=1}^{\infty}z_s^n(t) = \sum\nolimits_{n=1}^{\infty} |Z_s^n|\sin(n\omega t+\angle Z_s^n),\\
         &= \sum\nolimits_{n=1}^{\infty} |Z^n\mathcal{C}_s(n\omega)|\sin(n\omega t+\angle Z^n + \angle \mathcal{C}_s(n\omega)),\\
        % e_s(t) &= \sum\nolimits_{n=1}^{\infty}e_s^n(t) = \sum\nolimits_{n=1}^{\infty} |e_s^n|\sin(n\omega t+\angle e_s^n),\\
        % m(t) &= \sum\nolimits_{n=1}^{\infty}m^n(t) = \sum\nolimits_{n=1}^{\infty} |M^n|\sin(n\omega t+\angle M^n),\\
        % v(t) &= \sum\nolimits_{n=1}^{\infty}v^n(t) = \sum\nolimits_{n=1}^{\infty} |V^n|\sin(n\omega t+\angle V^n),\\
        u(t) &= \sum\nolimits_{n=1}^{\infty}u^n(t) = \sum\nolimits_{n=1}^{\infty} |U^n|\sin(n\omega t+\angle U^n),\\
        y(t) &= \sum\nolimits_{n=1}^{\infty}y^n(t) = \sum\nolimits_{n=1}^{\infty} |Y^n|\sin(n\omega t+\angle Y^n), 
    \end{aligned}
\end{equation}
where the phase for each signal, such as the $\angle E^n$, is defined within the range of $(-\pi,\pi]$. The Fourier transforms of the signals and their $n$th harmonic are denoted as $E(\omega)$ ($E^n(\omega)$), $Z(\omega)$ ($Z^n(\omega)$), $Z_s(\omega)$ ($Z_s^n(\omega)$), $U(\omega)$ ($U^n(\omega)$), and $Y(\omega)$ ($Y^n(\omega)$). 

% To address this, the closed-loop reset control system is initially configured to operate as a two-reset control system by ensuring it satisfies the two-reset condition outlined in \textcolor{red}{J2}. Once this condition is met, the closed-loop HOSIDFs developed in this study—encompassing the sensitivity function, complementary sensitivity function, and control sensitivity function for each harmonic—are applied to analyze the system dynamics.
In the sinusoidal-input frequency response analysis of closed-loop systems, two scenarios are identified: two-reset control systems, which undergo two resets per steady-state cycle, and multiple-reset control systems, which experience more than two resets per cycle. Multiple-reset actions are often indicative of high-magnitude higher-order harmonics that can impair performance and are therefore undesirable \cite{ZHANG2024106063}. Moreover, existing closed-loop SIDF analysis methods for reset control systems generally assume the operation of two-reset systems \cite{saikumar2021loop,van2024higher}. 

To this end, in the design of reset control systems, we apply the approach detailed in \cite{Xinxin_multiple_reset} to ensure that the system achieves two reset instants per steady-state cycle. This configuration ensures that the first-order harmonic \( z_s^1(t) \) dominates the reset-triggered signal \( z_s(t) \) as expressed in \eqref{eq: e,y,u}, while the contributions from higher-order harmonics \( z_s^n(t) \) for \( n > 1 \) are negligible. Based on this, we propose the following assumption:

\begin{assum}
\label{assum:2reset}
In the closed-loop reset control system with a sinusoidal input signal \(\sin(\omega t) \), the reset-triggered signal is given by \( z_s(t) = z_s^1(t) \).
\end{assum}

While this assumption may introduce some deviation in the closed-loop analysis, such deviations are expected to be minor, as will be demonstrated in the forthcoming examples. 

Under Assumption \ref{assum:2reset}, the set of reset instants of the closed-loop reset control system is given by $J := \{t_\eta = (\eta\pi - \angle Z_s^1)/\omega|\eta\in \mathbb{Z}^+\}$. Then, Theorem \ref{thm: closed-loop HOSIDF} introduces the HOSIDFs for the closed-loop two-reset control systems.

\begin{thm} 
\label{thm: closed-loop HOSIDF}
Consider a closed-loop two-reset control system in Fig. \ref{fig1:RC system}, with the input signal defined as \( r(t) = |R|\sin(\omega t) \), under Assumptions \ref{assum: closed-loop stability} and \ref{assum:2reset}. Utilizing the ``Virtual Harmonic Generator" approach \cite{nuij2006higher}, the input signal $r(t)$ generates harmonics \(r^n(t) = |R|\sin( n\omega t)\) with Fourier transforms of \( R^n(\omega) = |R|\mathscr{F}[\sin(n\omega t)] \). The $n$th Higher-Order Sinusoidal Input Sensitivity Function (HOSISF) \( \mathcal{S}_n(\omega) \), Higher-Order Sinusoidal Input Complementary Sensitivity Function \( \mathcal{T}_n(\omega) \), and the Higher-Order Sinusoidal Input Control Sensitivity Function \( \mathcal{CS}_n(\omega) \) are given as follows:
\begin{equation}
    \label{eq: sensitivity functions in CL}
    \begin{aligned}
    \mathcal{S}_n(\omega) &= \frac{E^n(\omega)}{R^n(\omega)}
    % \\&
    =\begin{cases}
	 {1}/{(1+\mathcal{L}_{o}(\omega))} , & \text{for } n=1,\\
     - \mathcal{S}_{l}(n\omega) \cdot|\mathcal{S}_1(\omega)|e^{jn\angle \mathcal{S}_1(\omega)} \cdot\Gamma(\omega)\mathcal{L}_n(\omega)\mathcal{C}_4(n\omega), & \text{for odd} \ n > 1,\\
	 0, & \text{for even} \ n \geqslant 2, \end{cases}
    \end{aligned}
\end{equation}
\begin{equation}
\label{eq: complementary sensitivity functions in CL}
\begin{aligned}
    \mathcal{T}_n(\omega) &= \frac{Y^n(\omega)}{R^n(\omega)}
    % \\ &
    =\begin{cases}
    {\mathcal{L}_{o}(\omega)}/{[\mathcal{C}_{4}(\omega)\cdot(1+\mathcal{L}_{o}(\omega))]}, & \text{for } n=1,\\
   \mathcal{S}_{l}(n\omega) \cdot|\mathcal{S}_1(\omega)|e^{jn\angle \mathcal{S}_1(\omega)} \cdot\Gamma(\omega) \mathcal{L}_n(\omega), & \text{for odd} \ n > 1,\\
    0, & \text{for even} \ n \geqslant 2, \end{cases}
    \end{aligned}
\end{equation}
\begin{equation}
\begin{aligned}
    \mathcal{CS}_n(\omega) &= \frac{U^n(\omega)}{R^n(\omega)}
    % \\ &
    =\begin{cases}
    {\mathcal{L}_{o}(\omega)}/{[\mathcal{C}_{4}(\omega)\cdot\mathcal{P}(\omega)\cdot(1+\mathcal{L}_{o}(\omega))]}, & \text{for } n=1,\\
   \mathcal{S}_{l}(n\omega) \cdot|\mathcal{S}_1(\omega)|e^{jn\angle \mathcal{S}_1(\omega)} \cdot\Gamma(\omega) \mathcal{L}_n(\omega)/\mathcal{P}(n\omega), & \text{for odd} \ n > 1,\\
    0, & \text{for even} \ n \geqslant 2, \end{cases}
    % &= {U^n(\omega)}/{R^n(\omega)} ={\mathcal{T}_n(\omega)}/{\mathcal{P}(n\omega)},
\end{aligned} 		
\end{equation}
Where
\begin{equation}
\label{eq:Cw1,Lalpha}
    \begin{aligned}
    % \mathcal{C}_{\psi}(n\omega) &= \mathcal{C}_{\phi}(n\omega)\mathcal{P}(n\omega),\\
    % \mathcal{L}_{\rho}(n\omega) &= \mathcal{C}_{\psi}(n\omega)\mathcal{C}_{4}(n\omega),\\
    \mathcal{S}_{l}(n\omega)&= 1/(1+\mathcal{L}_{l}(n\omega)),\\
    % \mathcal{C}_z^1(\omega) &= {\mathcal{C}_1(\omega)}/{(1+\mathcal{L}_{o}(\omega))},\\
    % \mathcal{C}_{\phi}(n\omega) &= \mathcal{L}_{n}(\omega)/\mathcal{P}(n\omega),\\
    \Psi_n(\omega) &= {|\mathcal{L}_{\rho}(n\omega)|}/{|1+\mathcal{L}_{l}(n\omega)|},\\
    \mathcal{L}_{o}(n\omega) &= \mathcal{L}_{n}(\omega)+(\Gamma(\omega)-1)\mathcal{L}_{\rho}(n\omega),\\
    % \mathcal{L}_{\alpha}(n\omega) &= \mathcal{C}_{\rho}^n(\omega)\mathcal{C}_{3}(n\omega)\mathcal{P}(n\omega)\mathcal{C}_{4}(n\omega),\\
    \Delta_c^1(\omega) &=|\Delta_l(\omega)|\sin(\angle\Delta_l(\omega) -\angle \mathcal{C}_s(\omega)),\\   
    \mathcal{L}_{\rho}(n\omega) &= \mathcal{C}_{\rho}^n(\omega)\mathcal{C}_3(n\omega)\mathcal{P}(n\omega)\mathcal{C}_4(n\omega)\mathcal{C}_1(n\omega),\\
    \mathcal{L}_{l}(n\omega) &= {[\mathcal{C}_{l}(n\omega)+\mathcal{C}_{2}(n\omega)]\mathcal{C}_{3}(n\omega)\mathcal{P}(n\omega)\mathcal{C}_{4}(n\omega)}\mathcal{C}_{1}(n\omega),\\  
    \Gamma(\omega) &= 1/(1-{\sum\nolimits_{n=3}^{\infty}\Psi_n(\omega)\Delta_c^n(\omega)}{/\Delta_c^1(\omega)}), n = 2k+1, k\in\mathbb{N},\\ 
    \Delta_c^n(\omega) &= -|\Delta_l(n\omega)| \sin(\angle \Delta_l(n\omega) +\angle\mathcal{L}_{\rho}(n\omega) - 
    % \\&\indent \indent
     \angle (1+\mathcal{L}_{l}(n\omega))-n\angle \mathcal{C}_s(\omega)),\text{ for } n>1.    
    \end{aligned}
\end{equation}
The function \(\mathcal{L}_n(\omega)\) is provided in \eqref{eq: H_hn}, while the functions \(\mathcal{C}_\rho^n(\omega)\) and \(\Delta_l(n\omega)\) are defined in \eqref{eq: Delta_l, Delta_x, Delta_c, Delta_q, C_rho_n}.
\end{thm}
\begin{proof}
    The proof is provided in \ref{Proof for Theorem closed-loop Sen}.
\end{proof}
% Theorems \ref{thm: open-loop HOSIDF} and \ref{thm: closed-loop HOSIDF} are integrated into a frequency-domain analysis toolbox for the open-loop and closed-loop reset control systems. Note that the closed-loop analysis in Theorem \ref{thm: closed-loop HOSIDF} is under Assumption \ref{assum:2reset}. In practice, we avoid the multiple-reset control system.

Following the derivation process outlined in Theorem \ref{thm: closed-loop HOSIDF} and its proof in \ref{Proof for Theorem closed-loop Sen}, Corollary \ref{cor: dis_input} presents the Higher-Order Sinusoidal Input Process Sensitivity Function \( \mathcal{PS}_n(\omega) \) for closed-loop reset control systems.
\begin{cor}
\label{cor: dis_input}
Consider a closed-loop two-reset control system in Fig. \ref{fig1:RC system}, with the disturbance input signal \( d(t) = |D|\sin(\omega t) \), under Assumptions \ref{assum: closed-loop stability} and \ref{assum:2reset}. Utilizing the ``Virtual Harmonic Generator" \cite{nuij2006higher}, the input signal $d(t)$ generates harmonics \(d^n(t) = |D|\sin( n\omega t)\) with Fourier transforms of \( D^n(\omega) = |D|\mathscr{F}[\sin(n\omega t)] \). The $n$th Higher-Order Sinusoidal Input Process Sensitivity Function \( \mathcal{PS}_n(\omega) \) is given as follows:  
\begin{equation}
    \label{eq: PS in CL}
    \begin{aligned}
    \mathcal{PS}_n(\omega) &= \frac{E^n(\omega)}{D^n(\omega)}
    % \\&
    =\begin{cases}
	 {-\mathcal{P}(\omega)\mathcal{C}_4(\omega)}/{(1+\mathcal{L}_{o}(\omega))} , & \text{for } n=1,\\
      -\mathcal{S}_{l}(n\omega) \cdot|\mathcal{PS}_1(\omega)|e^{jn\angle \mathcal{PS}_1(\omega)} \cdot\Gamma(\omega)\mathcal{L}_n(\omega)\mathcal{C}_4(n\omega), & \text{for odd} \ n > 1,\\
	 0, & \text{for even} \ n \geqslant 2. \end{cases}
    \end{aligned}
\end{equation}
% where 
% \begin{equation}
%  \mathcal{PS}_{l}(n\omega) = -\mathcal{S}_{l}(n\omega) \cdot \mathcal{P}(n\omega).
% \end{equation}
\end{cor}

Based on Theorem \ref{thm: closed-loop HOSIDF}, Remark \ref{cor: e,y,u,v} provides a method for calculating the steady-state trajectories of sinusoidal reference input in closed-loop reset control systems.
\begin{rem}
\label{cor: e,y,u,v}
Under Assumptions \ref{assum: closed-loop stability} and \ref{assum:2reset}, in a closed-loop reset control system in Fig. \ref{fig1:RC system} with a sinusoidal reference signal $r(t) = |R|\sin(\omega t)$, the steady-state error signal $e_r(t)$, output signal $y_r(t)$, and control input signal $u_r(t)$ are given by
% utilizing the \enquote{Virtual Harmonic Generator}, the input signal $r(t)$ generates $n$ harmonics $r_n(t) = |R|\sin(n\omega t)$, with Fourier transform $R_n(\omega)$. 
    \begin{equation}
        \begin{aligned}
            e_r(t) &= \sum\nolimits_{n=1}^{\infty} |R|\cdot|\mathcal{S}_n(\omega)|\sin(n\omega t + \angle \mathcal{S}_n(\omega)),\\
            y_r(t) &= \sum\nolimits_{n=1}^{\infty} |R|\cdot| \mathcal{T}_n(\omega)|\sin(n\omega t + \angle \mathcal{T}_n(\omega)),\\
            u_r(t) &=\sum\nolimits_{n=1}^{\infty}  |R|\cdot|\mathcal{CS}_n(\omega)|\sin(n\omega t + \angle \mathcal{CS}_n(\omega)).
            % v(t) &= \sum\nolimits_{n=1}^{\infty}\mathscr{F}^{-1}\left[\Theta_n(n\omega)R_n(\omega)\right].
        \end{aligned}
    \end{equation}
\end{rem}

Based on Corollary \ref{cor: dis_input}, Remark \ref{rem: ed} provides a method for calculating the steady-state error in a closed-loop reset control system when subjected to a sinusoidal disturbance input.
\begin{rem}
\label{rem: ed}
Under Assumptions \ref{assum: closed-loop stability} and \ref{assum:2reset}, the steady-state error signal \( e_d(t) \) of a closed-loop reset control system in Fig. \ref{fig1:RC system}, with a sinusoidal disturbance input \( d(t) = |D|\sin(\omega t) \), is given by: 
% utilizing the \enquote{Virtual Harmonic Generator}, the input signal $r(t)$ generates $n$ harmonics $r_n(t) = |R|\sin(n\omega t)$, with Fourier transform $R_n(\omega)$. 
    \begin{equation}
        \begin{aligned}
            e_d(t) &= \sum\nolimits_{n=1}^{\infty} |D|\cdot|\mathcal{PS}_n(\omega)|\sin(n\omega t + \angle \mathcal{PS}_n(\omega)).
            % y(t) &= \sum\nolimits_{n=1}^{\infty} |R|\cdot| \mathcal{T}_n(\omega)|\sin(n\omega t + \angle \mathcal{T}_n(\omega)),\\
            % u(t) &=\sum\nolimits_{n=1}^{\infty}  |R|\cdot|\mathcal{CS}_n(\omega)|\sin(n\omega t + \angle \mathcal{CS}_n(\omega)).
            % v(t) &= \sum\nolimits_{n=1}^{\infty}\mathscr{F}^{-1}\left[\Theta_n(n\omega)R_n(\omega)\right].
        \end{aligned}
    \end{equation}
\end{rem}
% \subsection{Toolbox Results \rom{3}: The Frequency Response Analysis for the Closed-loop reset control systems}
\subsection{Validation of the Closed-loop HOSIDFs}
This subsection uses illustrative examples and conducts simulations and experiments to validate the accuracy of Theorem \ref{thm: closed-loop HOSIDF} and Corollary \ref{cor: dis_input}.

The illustrative system is designed within the generalized structure shown in Fig. \ref{fig1:RC system}, with its parameters specified as follows: \(\mathcal{C}_1(s) = {(s/(150\pi) + 1)}/{(s/(3000\pi) + 1)}\), \(\mathcal{C}_s(s) = 1/(s/100 +1)\), the reset controller is built with a BLS system \(\mathcal{C}_l = 1/(s/(300\pi) + 1)\) with a reset value $\gamma=0$, \(\mathcal{C}_2(s) =\mathcal{C}_4(s) =  1\), \(\mathcal{C}_3(s) = 45 \cdot {(s/(300\pi) + 1)}/{(s/(30000\pi) + 1)} \cdot {(s+30\pi)/s} \cdot {(s/(130\pi) + 1)}/{(s/(699\pi) + 1)}  \cdot {1}/{(s/(3000\pi) + 1)}\), and the plant \(\mathcal{P}(s)\) is the precision motion stage given in \eqref{eq:P(s)}. The system has been verified to be both stable and convergent. Additionally, the two-reset condition outlined in \cite{Xinxin_multiple_reset} is applied to ensure that this reset control system, when subjected to sinusoidal inputs, exhibits two reset instants per steady-state cycle across the entire operating frequency range.

To validate the accuracy of Theorem \ref{thm: closed-loop HOSIDF}, let \(||e_r||_\infty/||r||_\infty\) and \(||u_r||_\infty/||r||_\infty\) denote the ratios of the \(\mathscr{L}_\infty\) norms of the steady-state error \(e_r\) and control input \(u_r\) to the sinusoidal reference input \(r = \sin(\omega t)\), respectively. Figures \ref{Ex2_pre_sim}(a) and (b) compare the values derived from simulations with those predicted by Theorem \ref{thm: closed-loop HOSIDF}. The results confirm that Theorem \ref{thm: closed-loop HOSIDF} accurately predicts system dynamics across the frequency range \( [1,1000] \) Hz. Similar to the open-loop HOSIDF analysis in Fig. \ref{fig:ex4_pe_nh_final}, prediction accuracy improves with the number of harmonics \(N_h\) considered in the computation. In this study, \(N_h = 100\) is used to ensure reliable predictions. 
% Readers can increase the number of harmonics in the calculation to achieve more accurate predictions.

Next, experimental validation of Theorem \ref{thm: closed-loop HOSIDF} is conducted. Figures \ref{Ex2_pre_sim}(c) and (d) compare the steady-state error \( e_r(t) \) and control input \( u_r(t) \) of the system under a reference input \( r(t) = 6 \times 10^{-7} \sin(400\pi t) \) [m], obtained from simulations, experimental measurements, and predictions based on Theorem \ref{thm: closed-loop HOSIDF}. The results demonstrate good agreement between the predictions and and simulation data. Minor discrepancies between the experimental and simulation results can be attributed to approximations in system identification and noise in the measurements.

\begin{figure}[htp]
    \centering
    % \captionsetup{singlelinecheck = false, format= hang, justification=justified, font=footnotesize, labelsep=space}
    \includegraphics[width=1\columnwidth]{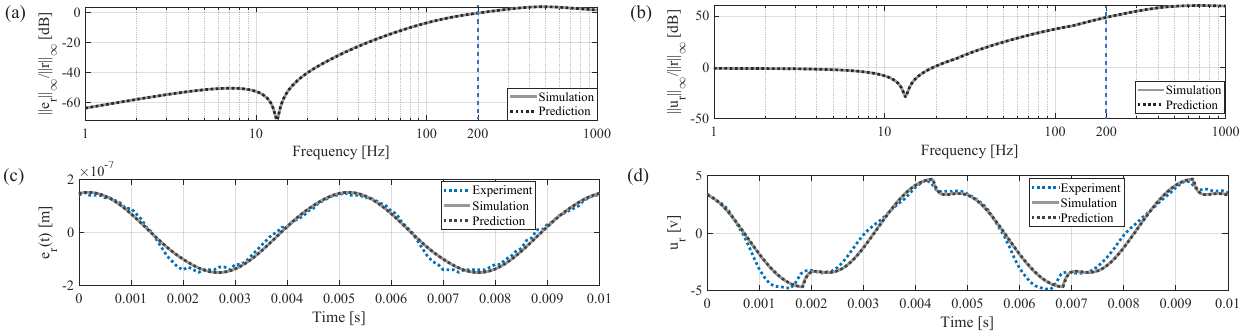}
    \caption{Theorem \ref{thm: closed-loop HOSIDF}-predicted and simulated values for (a) \(||e_r||_\infty/||r||_\infty\) and (b) \(||u_r||_\infty/||r||_\infty\) of the reset control system across the frequency range \([1, 1000]\) Hz. (c) Steady-state error signal \(e_r(t)\) and (d) control input signal \(u_r(t)\) for the system under the reference input \(r(t) = 6 \times 10^{-7} \sin(400\pi t)\) [m], as determined by Theorem \ref{thm: closed-loop HOSIDF} prediction, simulation, and experimental results.}
	\label{Ex2_pre_sim}
\end{figure}
% \begin{figure}[htp]
%     \centering
%     % \captionsetup{singlelinecheck = false, format= hang, justification=justified, font=footnotesize, labelsep=space}
%     \includegraphics[width=0.99\columnwidth]{figs/sim_pre_er_ed_max.pdf}
%     \caption{(a) Theorem \ref{thm: closed-loop HOSIDF}-predicted and simulated $||e_r||_\infty/||r||_\infty$ value. (b) Corollary \ref{cor: dis_input}-predicted and simulated $||e_d||_\infty/||d||_\infty$ value.}
% 	\label{sim_pre_sen_mag}
% \end{figure}

Similarly, the accuracy of Corollary \ref{cor: dis_input} is validated. Figure \ref{Ex2_pre_sim_d}(a) compares the \(||e_d||_\infty/||d||_\infty\) values derived from predictions and simulations. Figure \ref{Ex2_pre_sim_d}(b) compares the steady-state error \(e_d(t)\) of the system under a disturbance input \(d(t) = 1 \times 10^{-4}\sin(40\pi t)\) [m], obtained from predictions, simulations, and experiments. The results confirm that Corollary \ref{cor: dis_input} accurately predicts the system's response to sinusoidal disturbances.

% These variations are expected in real-world applications, where modeling imperfections and environmental disturbances can impact system behavior. 

% Overall, the consistency among the different data sources supports the validity of the HOSIDF method for analyzing closed-loop reset control systems and demonstrates its utility in practical control system design.

% This analysis underscores the effectiveness of the proposed methods in accurately predicting the behavior of the closed-loop reset control system under the specified sinusoidal input conditions.
 % assumptions made in Assumptions \ref{assum: before t1 base-linear} and \ref{assum:2reset}

\begin{figure}[htp]
    \centering
    % \captionsetup{singlelinecheck = false, format= hang, justification=justified, font=footnotesize, labelsep=space}
    \includegraphics[width=0.7\columnwidth]{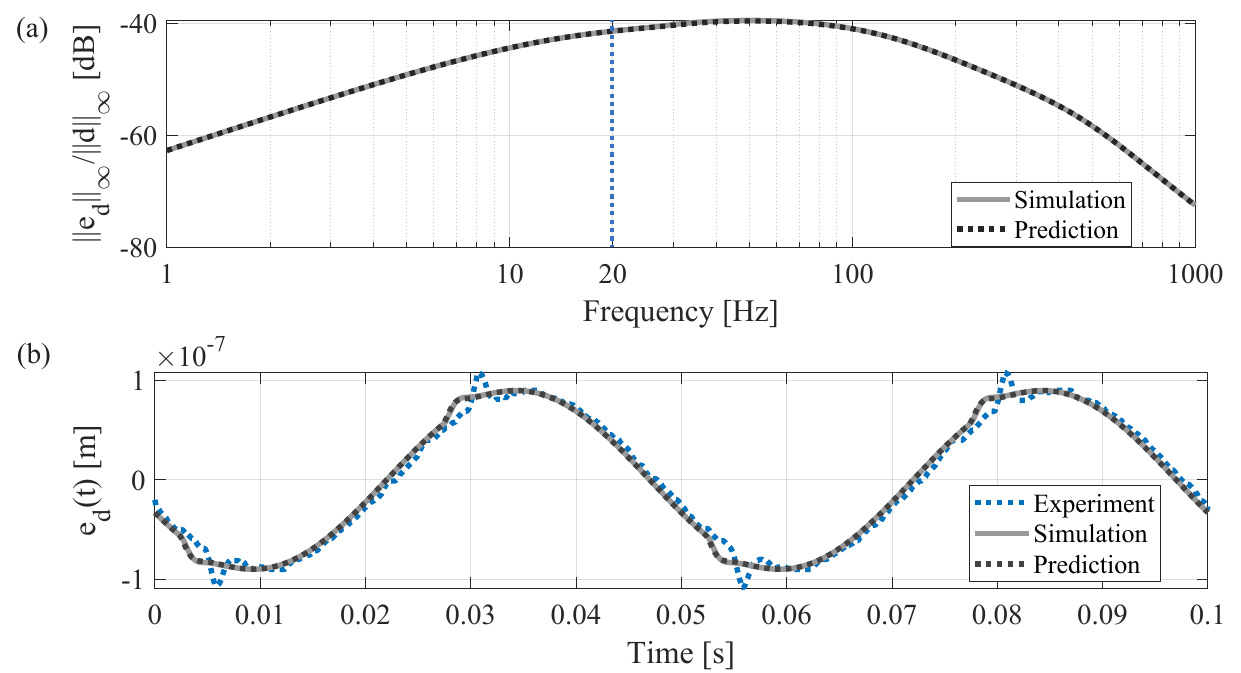}
    \caption{(a) Corollary \ref{cor: dis_input}-predicted and simulated \(||e_d||_\infty/||d||_\infty\) values of the reset control system across the frequency range \([1, 1000]\) Hz. (b) Comparison of Theorem \ref{thm: closed-loop HOSIDF}-predicted, simulated, and experimentally measured closed-loop steady-state error signal \(e_d(t)\) under the reference input signal $d(t) = 1 \times 10^{-4} \sin(40\pi t)$ [m].}
	\label{Ex2_pre_sim_d}
\end{figure}

After validating the accuracy, Theorem \ref{thm: closed-loop HOSIDF} and Corollary \ref{cor: dis_input} can be reliably employed to predict the behavior of closed-loop two-reset control systems. For illustration, Fig. \ref{validate_thm3_frequency_domain_sen} presents the Bode plots of the sensitivity function and the process sensitivity function for the closed-loop reset control system. The magnitude and phase information for each harmonic of the closed-loop reset control systems form the basis for analyzing system dynamics, such as reference tracking, and disturbance and noise rejection capabilities.
\begin{figure}[htp]
    \centering
    % \captionsetup{singlelinecheck = false, format= hang, justification=justified, font=footnotesize, labelsep=space}
    \includegraphics[width=1\columnwidth]{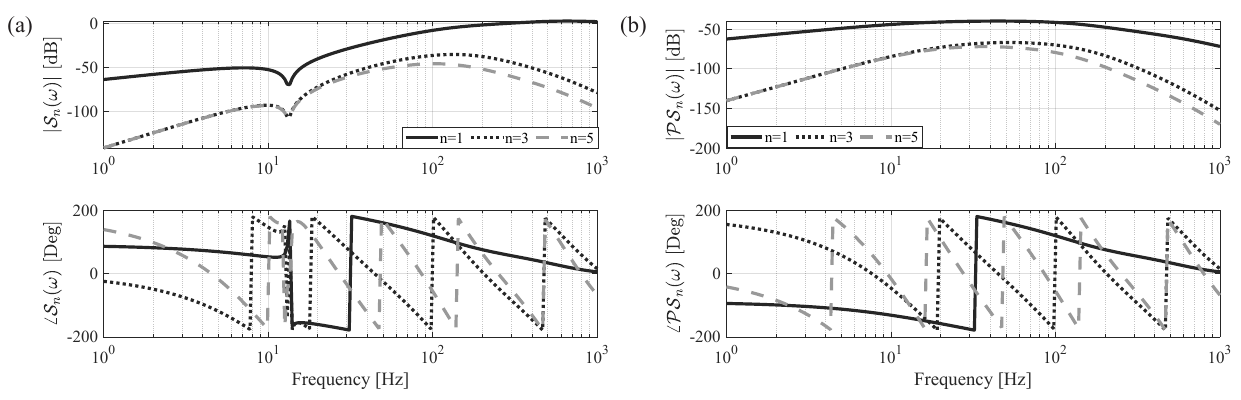}
    \caption{(a) The sensitivity function $\mathcal{S}_n(\omega)$ and (b) the process sensitivity function $\mathcal{PS}_n(\omega)$ of a closed-loop reset control system with $n=1, 3, 5$.}
	\label{validate_thm3_frequency_domain_sen}
\end{figure}

% \subsection{The Application of the frequency response analysis}
% Finally, 

% The application of this closed-loop analysis will be further discussed in Section \ref{sec: Application of Theorems}.
\section{Main Result 3: MATLAB App ``Reset Far" for Frequency Response Analysis of Generalized Reset Control Systems}
\label{sec:matlab app}
The HOSIDFs for open-loop and closed-loop generalized reset feedback control systems, depicted in Fig. \ref{fig1:RC system} and formulated in Theorems \ref{thm: open-loop HOSIDF} and \ref{thm: closed-loop HOSIDF}, have been integrated into a MATLAB application. The graphical user interface (GUI) of the App is shown in Fig. \ref{Reset_Far_App_figure}. It features five panels, each dedicated to specific functions as detailed below:
\begin{figure}[!t]
\centering
    % \captionsetup{singlelinecheck = false, format= hang, justification=justified, font=footnotesize, labelsep=space}
    \includegraphics[width=0.85\columnwidth]{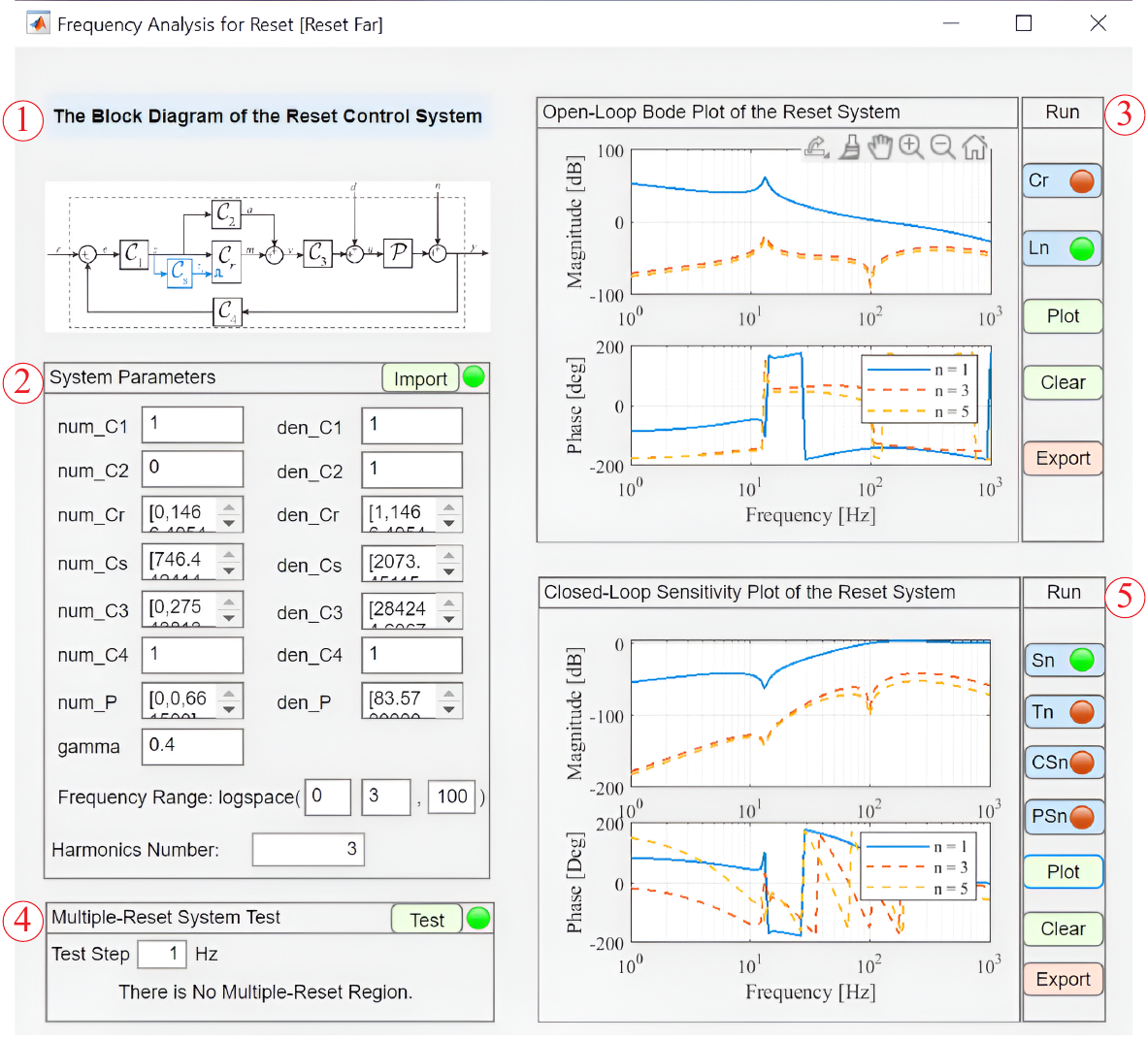}
    \caption{GUI of the frequency response analysis App for the generalized reset control system, named ``Reset Far".}
	\label{Reset_Far_App_figure}
\end{figure}
\begin{itemize}
    \item Panel \circled{1}: Displays the block diagram of the reset feedback control system in Fig. \ref{fig1:RC system}.
    \item Panel \circled{2}: Allows users to specify system parameters, including the numerators and denominators for systems \(\mathcal{C}_1\), \(\mathcal{C}_2\), \(\mathcal{C}_3\), \(\mathcal{C}_4\), \(\mathcal{C}_s\), \(\mathcal{C}_r\) (entered as the parameters of its base-linear counterpart \(\mathcal{C}_l\)), along with the reset value \(\gamma\), and the plant \(\mathcal{P}\). Additionally, the panel includes input fields for defining the frequency range for analysis (logarithmically spaced) and the number of harmonics to be considered.
    \item Panel \circled{3}: Select either ``Cr" or ``Ln" until the indicator turns green, then click the ``Plot" button. The HOSIDFs for the reset controller \(\mathcal{C}_r\) and the open-loop system \(\mathcal{L}_n(\omega)\), as derived from Theorem \ref{thm: open-loop HOSIDF}, will be displayed in this panel. Use the ``Clear" button to remove the plots, and the ``Export" button to save the HOSIDF data as a ``.mat" file to the workspace.
    \item Panel \circled{4}: Identifies the frequency range where the sinusoidal-input closed-loop reset control system exhibits multiple (more than two) reset instants per steady-state cycle based on the method in \cite{Xinxin_multiple_reset}. To use it, click the ``Test" button, which turns green when active, and select the sweeping step size, defaulting to 1 Hz. The output will either indicate ``There is No Multiple-Reset Region," meaning the system operates with only two reset instants per cycle across the tested frequency range, or it will specify ``Multiple-Reset Regions: \(f_\alpha\) to \(f_\beta\) [Hz]," showing the frequency range(s) where multiple resets occur, with \(f_\alpha,\ f_\beta\in\mathbb{R}^+\) as the boundaries. If multiple-reset regions are detected, subsequent closed-loop HOSIDF analysis may yield inaccuracies, and adjusting system design parameters is recommended until ``There is No Multiple-Reset Region" is achieved.
    \item Panel \circled{5}: Generates HOSIDFs for the closed-loop reset control system, including \(\mathcal{S}_n(\omega)\), \(\mathcal{T}_n(\omega)\), \(\mathcal{CS}_n(\omega)\), and \(\mathcal{PS}_n(\omega)\) based on Theorem \ref{thm: closed-loop HOSIDF} and Corollary \ref{cor: dis_input}. First, select ``Sn", ``Tn", ``CSn", or ``PSn" until the indicator turns green, then click ``Plot" to display. The ``Clear" button erases the plots, while ``Export" saves the HOSIDF data as a ``.mat" file to the workspace.
\end{itemize}
The App, along with detailed instructions to guide users through its usage, is accessible via the supplementary files.
% this \href{https://surfdrive.surf.nl/files/index.php/s/ytIcuh94QQjpVC2}{link}.

% The illustrative example is also presented in the following Section \ref{sec: Application of Theorems} for further clarification.

\section{Case Study: Utilizing the MATLAB App ``Reset Far" for Frequency-Domain Analysis of Reset Control Systems}
\label{sec: Application of Theorems}
% The proposed frequency response analysis methods in previous sections, specifically the HOSIDFs for both open-loop and closed-loop reset control systems, as detailed in Theorem \ref{thm: open-loop HOSIDF} and Theorem \ref{thm: closed-loop HOSIDF}. 
This section presents case studies to demonstrate the effectiveness of the proposed frequency response methods from Theorem \ref{thm: open-loop HOSIDF} and Theorem \ref{thm: closed-loop HOSIDF} in the frequency-domain analysis of reset control systems, applied to the precision motion stage \(\mathcal{P}(s)\) in \eqref{eq:P(s)}. 
% , and is designed as a two-reset control system, as verified using Theorem \ref{thm: Delta}.
\subsection{Frequency-Domain Analysis of Reset Control Systems}
\label{case study}
We design three control systems—PID, Constant-in-gain-Lead-in-phase (CgLp)-PID, and shaped CgLp-PID—as case studies. Note that these systems are primarily used to demonstrate the application of the proposed methods in system analysis, rather than representing optimized designs. The stability and convergence of the illustrative reset control system have been verified.

The CgLp reset element, as proposed in \cite{saikumar2019constant}, is composed of a First-Order Reset Element (FORE) and a lead element, as illustrated in Fig. \ref{fig:CsCgLp_structure}(a). 
The transfer function of the PID controller is defined as
\begin{equation}
    {\text{PID}} = k_p \biggl(1+\frac{\omega_i}{s}\biggr) \biggl(\frac{s/\omega_d+1}{s/\omega_t+1}\biggr)\biggl(\frac{1}{s/\omega_f+1}\biggr).
\end{equation}
\begin{figure}[!t]
\centering
    % \captionsetup{singlelinecheck = false, format= hang, justification=justified, font=footnotesize, labelsep=space}
    \includegraphics[width=0.6\columnwidth]{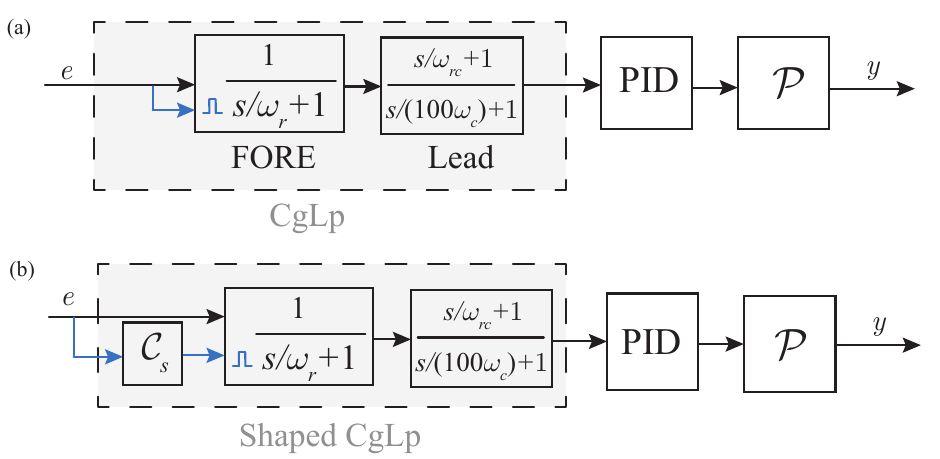}
    \caption{Block diagrams of the open-loop (a) CgLp-PID and (b) shaped CgLp-PID control systems.}
	\label{fig:CsCgLp_structure}
\end{figure}

By leveraging the phase lead advantage of reset control, the CgLp-PID element in this study is designed to provide phase lead while maintaining similar gain properties to a linear PID controller \cite{saikumar2019constant}. The cross-over frequency of \(\mathcal{L}_1(\omega)\) from \eqref{eq: H_hn}, where \(|\mathcal{L}_1(\omega)| = 0\) dB, is defined as the bandwidth of the open-loop system. The design parameters for the CgLp-PID control system are as follows: \( k_p = 35.7 \), \( \omega_c = 240\pi \) [rad/s], \( \omega_r = 244.8\pi \) [rad/s], \( \gamma=0 \), \( \omega_d = 120\pi \) [rad/s], \( \omega_t = 480\pi \) [rad/s], \( \omega_{rc} = 216\pi \) [rad/s], \( \omega_i = 24\pi \) [rad/s], and  \( \omega_f = 2400\pi \) [rad/s]. As shown in the Bode plots in Fig. \ref{fig:Ln_cglp_pid}, both the CgLp-PID and PID systems achieve a bandwidth of 120 Hz and maintain identical low-frequency gains. However, the CgLp-PID system achieves a phase margin of 40.7 degrees, providing a 15-degree improvement over the PID system's phase margin of 25.7 degrees.

\begin{figure}[!t]
\centering
    % \captionsetup{singlelinecheck = false, format= hang, justification=justified, font=footnotesize, labelsep=space}
    \includegraphics[width=0.73\columnwidth]{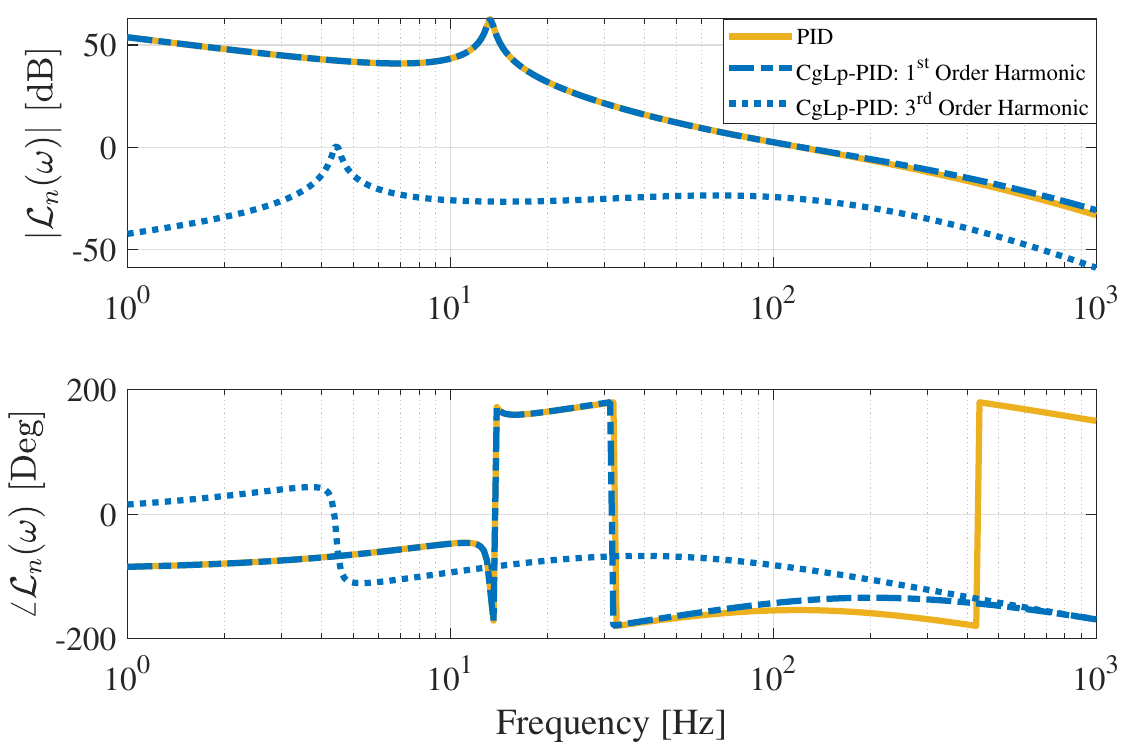}
    \caption{Bode plots for the PID control system alongside the first (\(n=1\)) and third (\(n=3\)) order HOSIDF \(\mathcal{L}_n\) for the CgLp-PID control systems.}
	\label{fig:Ln_cglp_pid}
\end{figure}

% However, the inherent nonlinearity of reset control introduces higher-order harmonics into the system. These harmonics disrupt the superposition principle, especially in multi-input closed-loop reset control systems. Moreover, when the higher-order harmonics have significant magnitude, they can increase the system's sensitivity to high-frequency noise, leading to degraded performance. The objective is to minimize these higher-order harmonics while retaining the performance benefits provided by the first-order harmonic. The newly proposed generalized reset control system in Fig. \ref{fig1:RC system} provide more tuning freedom this purpose. Since this study is not focus on the reset control system design, but introduce the application of the analysis tool to the system design.

A shaping filter \(\mathcal{C}_s(s)\) is designed and integrated into the CgLp-PID control system to form the shaped CgLp-PID control system, as shown in Fig. \ref{fig:CsCgLp_structure}(b). Note that in this case study, the shaping filter $\mathcal{C}_s(s)$ is specifically designed to reduce high-order harmonics of the CgLp-PID control system at the target frequency 100 Hz. By adjusting the parameters of \(\mathcal{C}_s(s)\), high-order harmonics at other targeted frequencies can be reduced as well. However, since this example primarily serves as an example to illustrate the application of the proposed frequency response analysis methods, the detailed design and tuning process of the shaping filter will be explored in future research. The transfer function of $\mathcal{C}_s(s)$ is given by
\begin{equation}
\label{EQ: CS_DE}
    \mathcal{C}_s(s) = \frac{s/(660\pi) + 1}{s/(237.6\pi) + 1}.
\end{equation}

Then, Theorems \ref{thm: open-loop HOSIDF} and \ref{thm: closed-loop HOSIDF} are employed to analyze and compare the frequency-domain characteristics of the PID, CgLp-PID, and shaped CgLp-PID control systems.

First, using Theorem \ref{thm: open-loop HOSIDF} (Panel \circled{3} in Fig. \ref{Reset_Far_App_figure}), the parameters of the shaped CgLp-PID system are tuned to \( \omega_r = 466.8\pi \) [rad/s] and \( \gamma = 0.4 \), ensuring the same phase margin as the CgLp-PID system. Figure \ref{Cscglp_Ln_final} shows the Bode plots of the PID control system along with the open-loop HOSIDF \(\mathcal{L}_n(\omega)\) for both the CgLp-PID and shaped CgLp-PID control systems, with \(n=1\) and \(n=3\). For simplicity, high-order harmonics (\(n > 3\)) are omitted in the figure, as they have lower magnitudes than the third-order harmonics but can be derived using Theorem \ref{thm: open-loop HOSIDF}.

As shown in Fig. \ref{Cscglp_Ln_final}, the shaped CgLp-PID system maintains the same phase margin as the CgLp-PID system while offering a larger bandwidth. Additionally, it effectively reduces high-order harmonics. Specifically, at an input frequency of 100 Hz, the magnitude of the third-order harmonic is decreased from 0.0592 in the CgLp-PID system to $9.14 \times 10^{-5}$ in the shaped CgLp-PID system, representing a reduction of 99.85\%. 
% This ongoing investigation aims to further refine the shaping techniques for optimal performance across a broader range of operating conditions.
% This tuning optimizes the system by leveraging the benefits of reset control while achieving both phase lead and enhanced gain characteristics.

% 
% Building upon the CgLp-PID control system in Fig. \ref{fig:CsCgLp_structure}(a), we propose a new shaped CgLp-PID control system by integrating a shaping filter \(\mathcal{C}_s\), as depicted in Fig. \ref{fig:CsCgLp_structure}(b). Then, these two systems are analyzed using the open-loop and closed-loop analysis methods developed in Theorems \ref{thm: open-loop HOSIDF} to \ref{thm: closed-loop HOSIDF}. 

% In the first-order harmonic analysis in Fig. \ref{Cscglp_Ln_final}, the shaped CgLp-PID system achieves a higher bandwidth compared to the CgLp-PID system.
\begin{figure}[!t]
\centering
    % \captionsetup{singlelinecheck = false, format= hang, justification=justified, font=footnotesize, labelsep=space}
    \includegraphics[width=0.8\columnwidth]{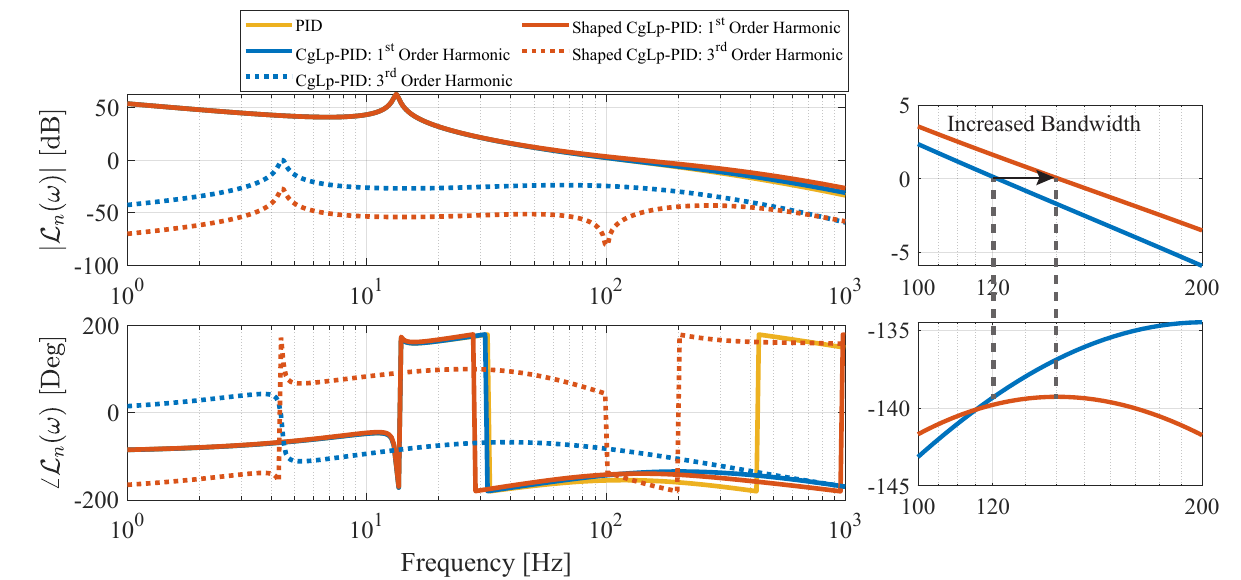}
    \caption{Bode plots of the PID system and open-loop HOSIDF \(\mathcal{L}_n\) for the CgLp-PID and shaped CgLp-PID control systems, with harmonics \(n=1\) and \(n=3\). The zoomed-in figure on the right highlights the first-order harmonic within the frequency range \([100, 300]\) Hz.}
	\label{Cscglp_Ln_final}
\end{figure}

Second, the multiple-reset control system identification tool \cite{Xinxin_multiple_reset} (Panel \circled{4} in Fig. \ref{Reset_Far_App_figure}) is applied to verify that both the sinusoidal-input CgLp-PID and shaped CgLp-PID control systems operate as two-reset control systems within the working frequency range of \([1,1000]\) Hz. This verification ensures that the two-reset condition is met for accurate closed-loop HOSIDF analysis.
% prevent the system from introducing undesirable high-magnitude high-order harmonics in the closed loop, which can adversely affect the overall system response.
\begin{figure}[!t]
    \centering
    \includegraphics[width=0.95\columnwidth]{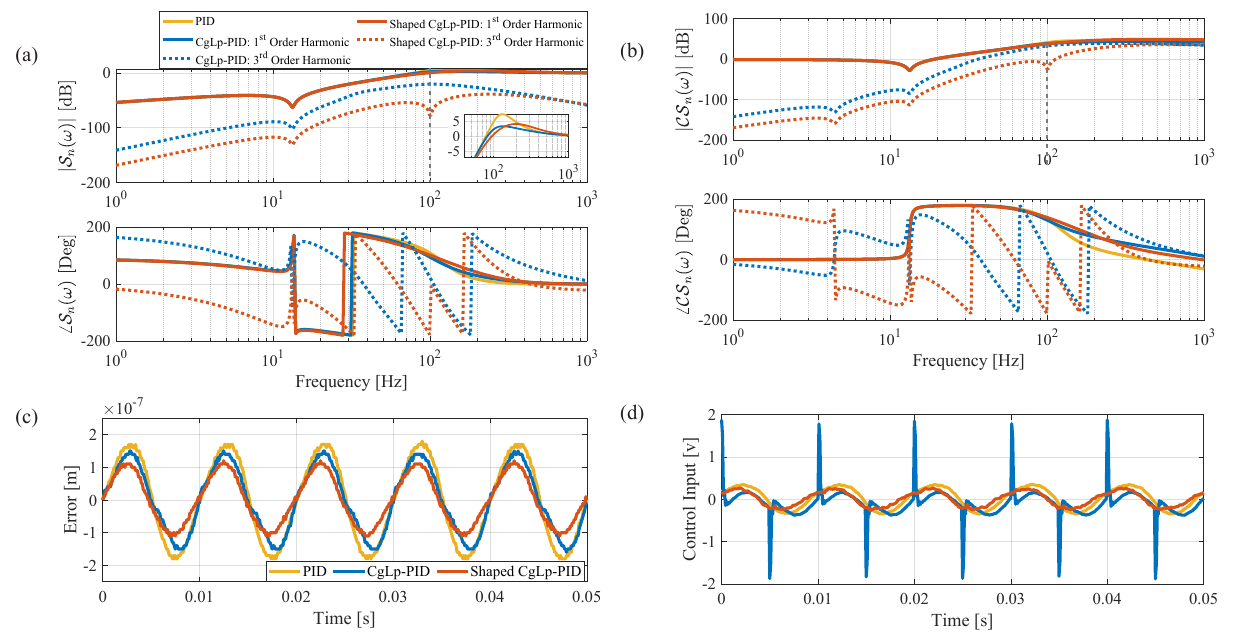}
    % \captionsetup{singlelinecheck = false, format= hang, justification=justified, font=footnotesize, labelsep=space}
    \caption{(a) The Higher-Order Sinusoidal Input Sensitivity Function (HOSISF) \(\mathcal{S}_n\) and (b) Control Sensitivity Function \(\mathcal{CS}_n(\omega)\) for the closed-loop PID, CgLp-PID, and shaped CgLp-PID systems, where $n=1,3$. (c) The experimentally measured steady-state errors and (d) control input signals for these three systems at the input frequency of 100 Hz.}
	\label{fig:CScglp_Sn_Csn}
\end{figure}

Third, Theorem \ref{thm: closed-loop HOSIDF} (Panel \circled{5} in Fig. \ref{Reset_Far_App_figure}) is applied to perform the closed-loop frequency response analysis for these three systems. Figures \ref{fig:CScglp_Sn_Csn}(a) and (b) show the sensitivity function \(\mathcal{S}_n\) and the control sensitivity function \(\mathcal{CS}_n\) for the PID control system with $n=1$, as well as for CgLp-PID and shaped CgLp-PID control systems, with \(n=1\) and \(n=3\).

From the analysis of the sensitivity function in Fig. \ref{fig:CScglp_Sn_Csn}(a), the CgLp-PID and shaped CgLp-PID control systems exhibit similar first-order harmonics. However, in the shaped CgLp-PID system, a reduction in the magnitude of high-order harmonics (\(n=3\)) is observed. Specifically, at an input frequency of 100 Hz, the value of $|\mathcal{S}_3(\omega)|$ decreases from 0.096 to $1.36 \times 10^{-4}$, which corresponds to a 99.86\% reduction. This decrease in the sensitivity function will result in a corresponding reduction in steady-state errors, as demonstrated by the subsequent experimental results.
% In the frequency range of $[200, 1000]$ Hz, the first-order harmonic in the shaped CgLp-PID control system is slightly higher than in the CgLp-PID system due to the trade-off associated with higher bandwidth, it can also be analyzed using Theorem \ref{thm: closed-loop HOSIDF}. 
% Nevertheless, the primary focus of the shaped CgLp-PID control system is on the analysis at the targeted frequency of 100 Hz.

Figures \ref{fig:CScglp_Sn_Csn}(c) and (d) present the experimentally measured steady-state error and control input signals for the three control systems when subjected to a sinusoidal input signal \(r(t) = 1.2 \times 10^{-7}\sin(200\pi t)\) [m]. For a quantitative analysis, Table \ref{tb: e_u_3systems} provides a summary of the maximum errors and control inputs for the three systems. Notably, the shaped CgLp-PID control system achieves a 21.43\% reduction in maximum error compared to the CgLp-PID control system at 100 Hz. This improvement in precision is primarily attributed to the reduction in $|\mathcal{S}_n(\omega)|$ at 100 Hz, as shown in Fig. \ref{fig:CScglp_Sn_Csn}(a).

Additionally, the advantages of reducing high-order harmonics in the shaped CgLp-PID system are more pronounced in the control input signal. The control sensitivity function analysis in Fig. \ref{fig:CScglp_Sn_Csn}(b) shows that the CgLp-PID system exhibits substantial high-magnitude high-order harmonics at 100 Hz, which are nearly equal to the first-order harmonic. This results in noticeable spikes in the control input signal, as observed in Fig. \ref{fig:CScglp_Sn_Csn}(d). In contrast, the shaped CgLp-PID system effectively reduces these high-order harmonics, leading to a smoother, more linear control input signal. As highlighted in Table \ref{tb: e_u_3systems}, the maximum control input required by the shaped CgLp-PID system is reduced by 85.64\% compared to the CgLp-PID system. 
\begin{table}[!t] 
\captionsetup{singlelinecheck = false, format= hang, justification=justified, font=footnotesize, labelsep=space}
\caption{Maximum steady-state errors \(||e||_\infty\) [m] and maximum control input \(||u||_\infty\) [v] for the PID, CgLp-PID and shaped CgLp-PID control systems.}
\label{tb: e_u_3systems}
\centering
    \renewcommand{\arraystretch}{1.5} % Adjust row height (optional)
\fontsize{8pt}{8pt}\selectfont
% \resizebox{0.48\columnwidth}{!}{
\begin{tabular}{|c|c|c|}
\hline
% \multirow{2}{*}{ Systems} & \multicolumn{5}{c|}{Frequency [Hz]}\\ \cline{2-6}
Systems & $||e||_\infty$ [m] & $||u||_\infty$ [v]\\ \hline
PID & $1.7\times 10^{-7}$ & 0.35 \\ %\hline
CgLp-PID & $1.4\times 10^{-7}$ & 1.88\\ %\hline
Shaped CgLp-PID & $1.1\times 10^{-7}$ & 0.27\\ 
\hline
\end{tabular}
% }
\end{table}

Results in Fig. \ref{fig:CScglp_Sn_Csn} show that the shaped CgLp-PID system not only improves steady-state accuracy but also reduces the actuation force, enhancing overall control efficiency.
% This advantage is further corroborated by the transient response comparison shown in Fig. \ref{fig:step_responses_final}. The results demonstrate that the shaped CgLp-PID system maintains comparable tracking performance while significantly reducing control input force, compared to the CgLp-PID system. In precision motion control applications, where there may be limits on actuation force to protect the device or avoid excessive disturbances, the efficient actuation of shaped CgLp-PID system will offer a clear advantage.
% The sensitivity function analysis in Fig. \ref{fig:CScglp_Sn_Csn}($a_1$) and ($b_1$) predicts the two system exhibit similar transient responses, while the shaped CgLp-PID system requires less actuation force. This prediction is validated by the experimental step responses for the two systems in Fig. \ref{fig:step_responses_final}. 
% , which demonstrate the effectiveness of the shaped CgLp-PID in improving performance while minimizing the control input force
% \begin{figure}[!t]
%     \captionsetup{singlelinecheck = false, format= hang, justification=justified, font=footnotesize, labelsep=space}
%     \includegraphics[width=0.48\columnwidth]{figs/step_responses_final.pdf}
%     \caption{(a) Experimental step responses of the CgLp-PID and shaped CgLp-PID systems. (b) The control input forces of these two systems.}
% 	\label{fig:step_responses_final}
% \end{figure}

\subsection{Discussions}
\label{case study}
The frequency response analysis methods in Theorems \ref{thm: open-loop HOSIDF} to \ref{thm: closed-loop HOSIDF}, developed for the generalized reset feedback control structure shown in Fig. \ref{fig1:RC system}, enable the tuning of linear elements in parallel, in series before and after the reset controller, and within the shaping filter to refine reset actions. This tuning flexibility broadens the potential for reset feedback control systems with enhanced performance. For example, ensuring that \({| \mathcal{S}_n(\omega)|}/{| \mathcal{S}_1(\omega)|} \to 0\) can preserve the advantages of first-order harmonics while suppressing higher-order harmonics to negligible levels. Striking this balance enables the application of the superposition principle in closed-loop, multiple-input reset control systems.
% as demonstrated in Remark \ref{rem: design_Sn}, by ensuring that \(\Gamma(\omega) \to 1\) in \eqref{eq:Cw1,Lalpha}, the closed-loop sensitivity function can be approximately analyzed using classical linear SIDF and loop-shaping methods. Additionally, from \eqref{eq: S_df, alpha}, 

Furthermore, in Section \ref{case study}, we proposed a shaped reset control structure that enhances tracking accuracy while reducing actuation demands at the targeted frequency. Industrial mechatronics applications often face tracking challenges due to dominant frequencies or specific disturbances with notable spectral characteristics, such as friction, vibrations, actuator dynamics, and sensor noise \cite{iwasaki2012high, chow2012disturbance, gruntjens2019hybrid}. The proposed shaped reset control structure is well-suited to address these challenges. However, this study primarily focuses on frequency response analysis; future research will explore detailed parameter optimization and tuning for targeted frequencies across different applications.

% Furthermore, the benefits of the generalized reset control structure extend beyond the specific case of shaped reset control. The tools presented in Theorems \ref{thm: open-loop HOSIDF} to \ref{thm: closed-loop HOSIDF} are designed to support this objective effectively.

% Moreover, these tools are implemented into a MATLAB app, allowing control engineers to intuitively utilize the tool for reset control design in the frequency domain.

% The advantages of reducing higher-order harmonics in reset control systems extend beyond the improved performance demonstrated in Fig. \ref{fig:CScglp_Sn_Csn}. 

% The analytical tools in Theorems \ref{thm: open-loop HOSIDF} through \ref{thm: closed-loop HOSIDF} provide both theoretical insights and practical methods for tuning and optimizing control systems for better performance. For instance, optimizing the shaping filter \(\mathcal{C}_s\) in the shaped CgLp-PID system (Figure \ref{fig:CsCgLp_structure}(b)) can enhance system performance. 
\section{Conclusion}
\label{sec: conclusion}
In conclusion, this study develops frequency response analysis tools for generalized single reset-state reset control systems, which are integrated into a MATLAB App, making them more accessible and intuitive for industrial applications. The first tool introduces Higher-Order Sinusoidal Input Describing Functions (HOSIDFs) for open-loop reset control systems, while the second provides HOSIDFs for the frequency response analysis of closed-loop reset control systems with two reset instants per steady-state cycle. Simulations and experiments validate the accuracy of the proposed methods.

Case studies conducted on a precision motion stage demonstrate the effectiveness of the proposed method in analyzing reset control systems. The frequency-domain analysis results show that the shaped CgLp-PID system reduces high-order harmonics in both open-loop and closed-loop configurations while maintaining similar first-order harmonics compared to the CgLp-PID control system. These frequency-domain advantages allow the shaped CgLp-PID system to achieve lower steady-state errors and reduced actuation force compared to previous PID and CgLp-PID control systems. Future research will focus on further exploring the design and tuning of this shaped control structure to optimize system performance.

The proposed tool is developed based on the method applied to a limited reset control structure from \cite{ZHANG2024106063}, but this study goes beyond a simple extension. The developed analysis methods establish a connection between open-loop and closed-loop analysis, enabling the fine-tuning of critical parameters such as the shaping filter \(\mathcal{C}_s\), linear controllers \(\mathcal{C}_1\), \(\mathcal{C}_2\), and \(\mathcal{C}_3\), which are arranged in series before, after, and parallel to the reset controller \(\mathcal{C}_r\), along with the feedback element \(\mathcal{C}_4\). This tuning capability facilitates the optimization of performance in reset feedback control systems.

Moreover, future research could leverage the frequency-domain techniques presented in Theorems \ref{thm: open-loop HOSIDF} and \ref{thm: closed-loop HOSIDF} to develop guidelines for designing reset control systems. These guidelines could focus on optimizing both transient and steady-state responses in reset-controlled mechatronics systems.

% While the closed-loop analysis discussed focuses on systems with reference input signals, the methods presented can also be adapted for systems with disturbance signals. This topic will be explored in future research to maintain the focus of the current study.

% Ex: \cite{Blondeletal2008,FabricioLiang2013} 
% \appendix

\section*{CRediT authorship contribution statement}
\textbf{Xinxin Zhang}: Conceptualization, Methodology, Formal analysis, Data curation, Software, Validation, Conducting Simulation and experiments, Writing – original draft, Writing – review and editing. \textbf{S. Hassan HosseinNia}: Conceptualization, Supervision, Resources, Review, and Editing. 

\section*{Declaration of competing interest}
The authors declare that they have no known competing financial interests or personal relationships that could have appeared to influence the work reported in this paper.
          
\section*{Acknowledgement}
Xinxin Zhang acknowledges the PhD grant from the China Scholarship Council.
                 
\section*{Declaration of Generative AI and AI-assisted technologies in the writing process}
During the preparation of this work the authors used [ChatGPT] in order to improve readability and language. After using this tool, the authors reviewed and edited the content as needed and take full responsibility for the content of the publication.

% \bibliographystyle{elsarticle-num-names} 
% \bibliography{References}

% \bibliographystyle{elsarticle-harv} 
\bibliographystyle{unsrt}
\bibliography{References}

\begin{thebibliography}{10}

\bibitem{schmidt2020design}
R~Munnig Schmidt, Georg Schitter, and Adrian Rankers.
\newblock {\em The design of high performance mechatronics-: high-Tech functionality by multidisciplinary system integration}.
\newblock Ios Press, 2020.

\bibitem{clegg1958nonlinear}
John~C Clegg.
\newblock A nonlinear integrator for servomechanisms.
\newblock {\em Transactions of the American Institute of Electrical Engineers, Part II: Applications and Industry}, 77(1):41--42, 1958.

\bibitem{guo2009frequency}
Yuqian Guo, Youyi Wang, and Lihua Xie.
\newblock Frequency-domain properties of reset systems with application in hard-disk-drive systems.
\newblock {\em IEEE Transactions on Control Systems Technology}, 17(6):1446--1453, 2009.

\bibitem{banos2012reset}
Alfonso Banos and Antonio Barreiro.
\newblock {\em Reset control systems}.
\newblock Springer, 2012.

\bibitem{krishnan1974synthesis}
KR~Krishnan and IM~Horowitz.
\newblock Synthesis of a non-linear feedback system with significant plant-ignorance for prescribed system tolerances.
\newblock {\em International Journal of Control}, 19(4):689--706, 1974.

\bibitem{horowitz1975non}
Isaac Horowitz and Patrick Rosenbaum.
\newblock Non-linear design for cost of feedback reduction in systems with large parameter uncertainty.
\newblock {\em International Journal of Control}, 21(6):977--1001, 1975.

\bibitem{banos2007definition}
Alfonso Ba{\~n}os and Angel Vidal.
\newblock Definition and tuning of a {PI+ CI} reset controller.
\newblock In {\em 2007 european control conference (ECC)}, pages 4792--4798. IEEE, 2007.

\bibitem{saikumar2019constant}
Niranjan Saikumar, Rahul~Kumar Sinha, and S~Hassan HosseinNia.
\newblock “constant in gain lead in phase” element--application in precision motion control.
\newblock {\em IEEE/ASME Transactions on Mechatronics}, 24(3):1176--1185, 2019.

\bibitem{van2020hybrid}
SJAM Van~den Eijnden, Marcel~Fran{\c{c}}ois Heertjes, WPMH Heemels, and Henk Nijmeijer.
\newblock Hybrid integrator-gain systems: A remedy for overshoot limitations in linear control?
\newblock {\em IEEE Control Systems Letters}, 4(4):1042--1047, 2020.

\bibitem{banos2007design}
Alfonso Banos and Angel Vidal.
\newblock Design of pi+ ci reset compensators for second order plants.
\newblock In {\em 2007 IEEE International Symposium on Industrial Electronics}, pages 118--123. IEEE, 2007.

\bibitem{villaverde2010reset}
Alejandro~Fern{\'a}ndez Villaverde, Antonio~Barreiro Blas, Joaquin Carrasco, and Alfonso~Ba{\~n}os Torrico.
\newblock Reset control for passive bilateral teleoperation.
\newblock {\em IEEE Transactions on Industrial Electronics}, 58(7):3037--3045, 2010.

\bibitem{guo2010optimal}
Yuqian Guo, Youyi Wang, Lihua Xie, Hui Li, and Weihua Gui.
\newblock Optimal reset law design and its application to transient response improvement of hdd systems.
\newblock {\em IEEE transactions on control systems technology}, 19(5):1160--1167, 2010.

\bibitem{heertjes2016experimental}
MF~Heertjes, KGJ Gruntjens, SJLM Van~Loon, N~Van~de Wouw, and WPMH Heemels.
\newblock Experimental evaluation of reset control for improved stage performance.
\newblock {\em IFAC-PapersOnLine}, 49(13):93--98, 2016.

\bibitem{zhao2019overcoming}
Guanglei Zhao, Dragan Ne{\v{s}}i{\'c}, Ying Tan, and Changchun Hua.
\newblock Overcoming overshoot performance limitations of linear systems with reset control.
\newblock {\em Automatica}, 101:27--35, 2019.

\bibitem{karbasizadeh2022continuous}
Nima Karbasizadeh and S~Hassan HosseinNia.
\newblock Continuous reset element: Transient and steady-state analysis for precision motion systems.
\newblock {\em Control Engineering Practice}, 126:105232, 2022.

\bibitem{lumkes2001control}
John~H Lumkes~Jr.
\newblock {\em Control strategies for dynamic systems: design and implementation}.
\newblock CRC Press, 2001.

\bibitem{grassi2001integrated}
Elena Grassi, Kostas~S Tsakalis, Sachi Dash, Sujit~V Gaikwad, Ward MacArthur, and Gunter Stein.
\newblock Integrated system identification and pid controller tuning by frequency loop-shaping.
\newblock {\em IEEE Transactions on Control Systems Technology}, 9(2):285--294, 2001.

\bibitem{richter2011advanced}
Hanz Richter.
\newblock {\em Advanced control of turbofan engines}.
\newblock Springer Science \& Business Media, 2011.

\bibitem{nuij2006higher}
PWJM Nuij, OH~Bosgra, and Maarten Steinbuch.
\newblock Higher-order sinusoidal input describing functions for the analysis of non-linear systems with harmonic responses.
\newblock {\em Mechanical Systems and Signal Processing}, 20(8):1883--1904, 2006.

\bibitem{saikumar2021loop}
Niranjan Saikumar, Kars Heinen, and S~Hassan HosseinNia.
\newblock Loop-shaping for reset control systems: A higher-order sinusoidal-input describing functions approach.
\newblock {\em Control Engineering Practice}, 111:104808, 2021.

\bibitem{karbasizadeh2022band}
Nima Karbasizadeh, Ali~Ahmadi Dastjerdi, Niranjan Saikumar, and S~Hassan HosseinNia.
\newblock Band-passing nonlinearity in reset elements.
\newblock {\em IEEE Transactions on Control Systems Technology}, 31(1):333--343, 2022.

\bibitem{van2024higher}
Luke~F van Eijk, Dragan Kosti{\'c}, Mohammad Khosravi, and S~Hassan HosseinNia.
\newblock Higher-order sinusoidal-input describing function analysis for a class of multiple-input multiple-output convergent systems.
\newblock {\em IEEE Transactions on Automatic Control}, 2024.

\bibitem{ZHANG2024106063}
Xinxin Zhang, Marcin~B. Kaczmarek, and S.~Hassan HosseinNia.
\newblock Frequency response analysis for reset control systems: Application to predict precision of motion systems.
\newblock {\em Control Engineering Practice}, 152:106063, 2024.

\bibitem{Xinxin_multiple_reset}
Xinxin Zhang and S.~Hassan HosseinNia.
\newblock Enhancing the reliability of closed-loop describing function analysis for reset control applied to precision motion systems.
\newblock {\em arXiv preprint \href{https://doi.org/10.48550/arXiv.2412.00502}{\color{gray}{arXiv:2412.00502}}}, 2024.

\bibitem{guo2019stability}
Yuqian Guo and Yanying Chen.
\newblock Stability analysis of delayed reset systems with distributed state resetting.
\newblock {\em Nonlinear Analysis: Hybrid Systems}, 31:265--274, 2019.

\bibitem{karbasizadeh2021fractional}
Nima Karbasizadeh, Niranjan Saikumar, and S~Hassan HosseinNia.
\newblock Fractional-order single state reset element.
\newblock {\em Nonlinear Dynamics}, 104:413--427, 2021.

\bibitem{pavlov2006uniform}
Alexey Pavlov, Nathan Van De~Wouw, and Hendrik Nijmeijer.
\newblock {\em Uniform output regulation of nonlinear systems: a convergent dynamics approach}, volume 205.
\newblock Springer, 2006.

\bibitem{pavlov2007frequency}
Alexey Pavlov, Nathan van~de Wouw, and Henk Nijmeijer.
\newblock Frequency response functions for nonlinear convergent systems.
\newblock {\em IEEE Transactions on Automatic Control}, 52(6):1159--1165, 2007.

\bibitem{dastjerdi2022closed}
Ali~Ahmadi Dastjerdi, Alessandro Astolfi, Niranjan Saikumar, Nima Karbasizadeh, Duarte Valerio, and S~Hassan HosseinNia.
\newblock Closed-loop frequency analysis of reset control systems.
\newblock {\em IEEE Transactions on Automatic Control}, 68(2):1146--1153, 2022.

\bibitem{barabanov2001bohl}
EA~Barabanov and AV~Konyukh.
\newblock Bohl exponents of linear differential systems.
\newblock {\em Mem. Differential Equations Math. Phys}, 24:151--158, 2001.

\bibitem{beker2000forced}
Orhan Beker, CV~Hollot, and Yossi Chait.
\newblock Forced oscillations in reset control systems.
\newblock In {\em Proceedings of the 39th IEEE conference on decision and control (Cat. No. 00CH37187)}, volume~5, pages 4825--4826. IEEE, 2000.

\bibitem{beker2004fundamental}
Orhan Beker, CV~Hollot, Yossi Chait, and Huaizhong Han.
\newblock Fundamental properties of reset control systems.
\newblock {\em Automatica}, 40(6):905--915, 2004.

\bibitem{pavlov2005convergent}
Alexey Pavlov, Nathan van~de Wouw, and Henk Nijmeijer.
\newblock Convergent piecewise affine systems: analysis and design part i: continuous case.
\newblock In {\em Proceedings of the 44th IEEE Conference on Decision and Control}, pages 5391--5396. IEEE, 2005.

\bibitem{iwasaki2012high}
Makoto Iwasaki, Kenta Seki, and Yoshihiro Maeda.
\newblock High-precision motion control techniques: A promising approach to improving motion performance.
\newblock {\em IEEE Industrial Electronics Magazine}, 6(1):32--40, 2012.

\bibitem{chow2012disturbance}
Hoi-Wai Chow and Norbert~C Cheung.
\newblock Disturbance and response time improvement of submicrometer precision linear motion system by using modified disturbance compensator and internal model reference control.
\newblock {\em IEEE Transactions on Industrial Electronics}, 60(1):139--150, 2012.

\bibitem{gruntjens2019hybrid}
KGJ Gruntjens, Marcel~Fran{\c{c}}ois Heertjes, SJLM Van~Loon, Nathan Van De~Wouw, and WPMH Heemels.
\newblock Hybrid integral reset control with application to a lens motion system.
\newblock In {\em 2019 American Control Conference (ACC)}, pages 2408--2413. IEEE, 2019.

\end{thebibliography}
\appendix

\section{The Proof of Theorem \ref{thm: open-loop HOSIDF}}
\label{Proof: thmCv}
\vspace*{12pt}
\begin{proof}
Consider an open-loop reset control system with an input \(e_o(t) = |E| \sin(\omega t + \angle E)\) and output \(y_o(t)\) as depicted in Fig. \ref{fig: open_loop_rcs}, satisfying Assumption \ref{assum: open-loop stability}. This proof derives the HOSIDFs for the reset controller \(\mathcal{C}_r\) and the open-loop system. The derivation process proceeds sequentially from the input signal \(e_o(t)\) on the left to the output signal \(y_o(t)\) on the right in Fig. \ref{fig: open_loop_rcs}.

First, the block \(\mathcal{C}_1\) receives the input \(e_o(t)\) and generates the output signal \(z_o(t)\). Let \(E_o(\omega)\) denotes the Fourier transform of \(e_o(t)\). Then, the output signal \(z_o(t)\) and its Fourier transform \(Z_o(\omega)\) are given by:
\begin{equation}
\label{eq: Wo(w)}
\begin{aligned}
 z_o(t) &=  |E\mathcal{C}_1(\omega)| \sin(\omega t + \angle E + \angle \mathcal{C}_1(\omega)),\\
Z_o(\omega) &= E_o(\omega)\mathcal{C}_1(\omega).   
\end{aligned}
\end{equation}
Next, the signal \( z_o(t) \) in \eqref{eq: Wo(w)} is filtered by the block \(\mathcal{C}_2\), producing the output signal \( a_o(t) \). The signal \( a_o(t) \) and its Fourier transform \( A_o(\omega) \) are given by the following equations:
\begin{equation}
\label{eq: Ao(w)0}
\begin{aligned}
a_o(t) &=  |E\mathcal{C}_1(\omega)\mathcal{C}_2(\omega)| \sin(\omega t + \angle E + \angle \mathcal{C}_1(\omega)+ \angle \mathcal{C}_2(\omega)),\\
A_o(\omega) &=Z_o(\omega)\mathcal{C}_2(\omega) = E_o(\omega)\mathcal{C}_1(\omega)\mathcal{C}_2(\omega).
\end{aligned}
\end{equation}
Meanwhile, the signal \(z_o(t)\) in \eqref{eq: Wo(w)} and the reset trigger signal \(z_s(t)\) are processed by the reset controller block \(\mathcal{C}_r\), producing the nonlinear output signal \(m_o(t)\). The reset trigger signal is given by:
\[
z_s(t) = |E \mathcal{C}_1(\omega) \mathcal{C}_s(\omega)| \sin\big(\omega t + \angle E + \angle \mathcal{C}_1(\omega) + \angle \mathcal{C}_s(\omega)\big).
\]
To conduct the HOSIDF analysis for a nonlinear system, the ``Virtual Harmonics Generator” \cite{nuij2006higher} is applied to decompose the input signal \(z_o(t)\) into its harmonic components \(z_o^n(t)\), expressed as:
\begin{equation}
\label{eq: won(t)_def}
    z_o^n(t) = |E\mathcal{C}_1(\omega)| \sin(n\omega t + n\angle E + n\angle \mathcal{C}_1(\omega)),
\end{equation}
where $z_o^1(t) = z_o(t)$. The Fourier transform of $z_o^n(t)$ is given by
\begin{equation}
\label{eq: w_on(t)2}
Z_o^n(\omega) = E_o(\omega)\mathcal{C}_1(\omega)e^{j(n-1)(\angle\mathcal{C}_1(\omega)+\angle E)}.    
\end{equation}
Then, define \(Z_o(\omega)\) and \(M_o(\omega)\) as the Fourier transforms of \(z_o(t)\) and $m_o(t)$, respectively. Referring \cite{ZHANG2024106063}, $M_o(\omega)$ is given by
\begin{equation}
\label{eq: mo(w)=ml+mnl}
\begin{aligned}
M_o(\omega) &= M_{l}(\omega) + M_{\rho}(\omega),
\end{aligned} 
\end{equation}
where
\begin{equation}
\label{eq: ml(w)+mnl(w)}
\begin{aligned}
M_{l}(\omega) &= Z_o(\omega)\mathcal{C}_{l}(\omega),\\
Q^n(\omega) &= 2Z_o^n(\omega)\Delta_q(\omega)e^{jn\angle \mathcal{C}_s(\omega)}/(n\pi),\\
M_{\rho}(\omega) &= \sum\nolimits_{n=1}^{\infty} \Delta_x(n\omega)Q^n(\omega), n=2k+1,k\in\mathbb{N}.
% M_{\rho}^n(\omega) &=\Delta_x(n\omega)Q^n(\omega),\\
\end{aligned} 
\end{equation}
From \eqref{eq: mo(w)=ml+mnl}, $M_o(\omega)$ can be expressed as the sum of its harmonics, expressed as
\begin{equation}
\label{eq: mo= m0n}
M_o(\omega) =  \sum\nolimits_{n=1}^{\infty} M_o^n(\omega),
\end{equation}
where
\begin{equation}
\label{eq: m_on(w)}
M_o^n(\omega)=
\begin{cases}
   M_{l}(\omega) + \Delta_x(\omega)Q^1(\omega), & \text{for } n=1,\\
    \Delta_x(n\omega)Q^n(\omega), & \text{for odd } n>1,\\
   0, & \text{for even } n \geq 2.
    \end{cases}
\end{equation}
From \eqref{eq: mo(w)=ml+mnl} and \eqref{eq: m_on(w)}, the HOSIDF of the reset controller \(\mathcal{C}_r\), describing the transfer function from \(Z_o^n(\omega)\) to \(M_o^n(\omega)\), is given by
\begin{equation}
\label{eq: Cr_hn} 
\begin{aligned}
\mathcal{C}_r^n(\omega) = \frac{M_o^n(\omega)}{Z_o^n(\omega)}=
\begin{cases}
        \mathcal{C}_{l}(\omega) + \mathcal{C}_{\rho}^1(\omega), & \text{for odd }n=1 \\
	\mathcal{C}_{\rho}^n(\omega), & \text{for odd }n>1,\\
		0,&\text{for even }n\geqslant 2,
  \end{cases}    
\end{aligned}
\end{equation}
where 
\begin{equation}
    \begin{aligned}
    \mathcal{C}_{\rho}^n(\omega) &=  2\Delta_x(n\omega)\Delta_q(\omega) e^{jn\angle \mathcal{C}_s(\omega)}/(n\pi).
    % \Delta_v(\omega) &= (I+e^{A_R\pi/\omega})(A_\rho e^{A_R\pi/\omega}-I)^{-1}(I-A_\rho).
    \end{aligned}
\end{equation}

The subsequent content builds upon the HOSIDF for the reset controller \(\mathcal{C}_r\) in \eqref{eq: Cr_hn} to derive the HOSIDF in \eqref{eq: H_hn} for the open-loop reset system shown in Fig. \ref{fig: open_loop_rcs}. This derivation follows the same logical framework used in \cite{van2024higher}, but to accommodate the reset control structure presented in this study and ensure the completeness of the proof for \eqref{eq: H_hn}, the following content is provided.

% Then, building upon the HOSIDF for the reset controller \(\mathcal{C}_r\) in \eqref{eq: Cr_hn}, and by employing the \enquote{Virtual Harmonics Generator} concept introduced in \cite{nuij2006higher}, while following the same logical framework used in \cite{saikumar2021loop} and \cite{van2024higher}, the HOSIDF for the open-loop reset control system in Fig. \ref{fig: open_loop_rcs} is derived in \eqref{eq: H_hn}. The proof is thus concluded.

From Fig. \ref{fig: open_loop_rcs}, the signal $v_o(t)$ is given by
\begin{equation}
\label{eq: vo(t)}
v_o(t) = a_o(t) + m_o(t).
\end{equation}

Define $V_o(\omega)$ as the Fourier transform of $v_o(t)$. Then, from \eqref{eq: Ao(w)0}, \eqref{eq: mo= m0n}, \eqref{eq: Cr_hn}, and \eqref{eq: vo(t)}, $V_o(\omega)$ is given by
\begin{equation}
\label{eq: Vo(w)}
\begin{aligned}
 V_o(\omega) &= A_o(\omega) + M_o(\omega)
 % \\&
 = E_o(\omega)\mathcal{C}_1(\omega)\mathcal{C}_2(\omega) + \sum\nolimits_{n=1}^{\infty} Z_o^n(\omega) \mathcal{C}_r^n(\omega).
\end{aligned}
\end{equation}
Substituting $Z_o^n(\omega)$ from \eqref{eq: w_on(t)2} into \eqref{eq: Vo(w)}, $V_o(\omega)$ is given by
\begin{equation}
\label{eq: Vo(w)2}
\begin{aligned}
V_o(\omega) = E_o(\omega)\mathcal{C}_1(\omega)\mathcal{C}_2(\omega) + \sum\nolimits_{n=1}^{\infty} E_o(\omega)\mathcal{C}_1(\omega)
% \cdot\\&
e^{j(n-1)(\angle\mathcal{C}_1(\omega)+\angle E)} \mathcal{C}_r^n(\omega).    
\end{aligned}
\end{equation}
From \eqref{eq: Vo(w)2}, \(V_o(\omega)\) can be written as the sum of its harmonics, denoted by \(V_o^n(\omega)\), expressed as:
\begin{equation}
\label{eq: Vo(w)3}
V_o(\omega) = \sum\nolimits_{n=1}^{\infty} V_o^n(\omega),
\end{equation}
where
\begin{equation}
\label{eq: v_on}
\begin{aligned}
& V_o^n(\omega) =
% \\&
\begin{cases}
   E_o(\omega)\mathcal{C}_1(\omega)[\mathcal{C}_{r}^1(\omega) + \mathcal{C}_2(\omega)], & \text{ for } n=1,\\
   E_o(\omega)\mathcal{C}_1(\omega)e^{j(n-1)(\angle\mathcal{C}_1(\omega)+\angle E)} \mathcal{C}_r^n(\omega), & \text{ for odd } n>1,\\
   0, & \text{ for even } n \geq 2.
    \end{cases}
\end{aligned}
\end{equation}
% As shown in Fig. \ref{fig1:OL_Block_Diagram}(b), 
The output signal \( y_o(t) \) of the open-loop system shown in Fig. \ref{fig: open_loop_rcs} exhibits nonlinear behavior and comprises \( n \) harmonic components, defined as \( y_o(t) = \sum\nolimits_{n=1}^{\infty} y_o^n(t) \). Let \(Y_o(\omega)\) and \(Y_o^n(\omega)\) denote the Fourier transforms of \(y_o(t)\) and \(y_o^n(t)\), respectively. Based on Fig. \ref{fig: open_loop_rcs}, $Y_o(\omega)$ and $Y_o^n(\omega)$ are given by
% the relationship between \(Y_o^n(\omega)\) and \(V_o^n(\omega)\) is given by:
\begin{equation}
\label{eq: yno,vno}
\begin{aligned}
  Y_o(\omega) &= \sum\nolimits_{n=1}^{\infty} Y_o^n(\omega),\\
  Y_o^n(\omega) &= V_o^n(\omega) \mathcal{C}_{3}(n\omega)\mathcal{P}(n\omega).
\end{aligned}
\end{equation}
By employing the \enquote{Virtual Harmonics Generator} \cite{nuij2006higher}, the input signal $e_o(t) = |E|\sin(\omega t+ \angle E)$ generates $n$ harmonics given by
\begin{equation}
\label{eq: en(t)}
e_o^n(t) = |E|\sin(n\omega t+ n\angle E).
\end{equation}
Define $E_o(w)$ and $E_o^n(w)$ as the Fourier transforms of $e_o(t)$ and $e_o^n(t)$, respectively. From \eqref{eq: en(t)}, $E_o^n(\omega)$ is given by 
\begin{equation}
\label{eq: en(w)}
E_o^n(w) = E_o(w)e^{j (n-1)\angle E}.
\end{equation}
Thus, from \eqref{eq: Cr_hn}, \eqref{eq: v_on}, \eqref{eq: yno,vno}, and \eqref{eq: en(w)}, the \(n\)-th transfer function \(\mathcal{L}_n(\omega)\) for the open-loop reset control system is given by
\begin{equation}
\label{Hn_def}
\mathcal{L}_n(\omega) = \frac{Y_o^n(\omega)}{E_o^n(\omega)} 
% \mathcal{C}_v^n(\omega) \mathcal{C}_{3}(n\omega)\mathcal{P}(n\omega). 
=\begin{cases}
\mathcal{C}_1(\omega)[\mathcal{C}_l(\omega) +\mathcal{C}_{\rho}^1(\omega) +  \mathcal{C}_2(\omega)]\mathcal{C}_{3}(\omega)\mathcal{P}(\omega), & \text{for } n=1, \\
\mathcal{C}_1(\omega)e^{j(n-1)\angle\mathcal{C}_1(\omega)}\mathcal{C}_{\rho}^n(\omega)\mathcal{C}_{3}(n\omega)\mathcal{P}(n\omega), & \text{for odd }n>1,\\
0,&\text{for even }n\geqslant 2.
  \end{cases}  
\end{equation}
Finally, equations \eqref{eq: Cr_hn} and \eqref{Hn_def} complete the proof.
\end{proof}
\section{The Proof of Theorem \ref{thm: closed-loop HOSIDF}}
\label{Proof for Theorem closed-loop Sen}
\vspace*{12pt}
\begin{proof}
Consider a closed-loop reset control system, as shown in Fig. \ref{fig1:RC system}, with a sinusoidal reference input signal \( r(t) = |R|\sin(\omega t) \), under Assumptions \ref{assum: closed-loop stability} and \ref{assum:2reset}. This proof derives the HOSIDFs for the closed-loop system in three steps. The approach follows a similar process to the proof in \cite{ZHANG2024106063}, but while the latter is limited to a system with \(\mathcal{C}_1=\mathcal{C}_s = \mathcal{C}_3 = \mathcal{C}_4 = 1\) and \(\mathcal{C}_2 = 0\), this proof is adapted to a more generalized structure as depicted in Fig. \ref{fig1:RC system}.\\
\textbf{Step 1: Construct the Block Diagram for Sinusoidal-Input Closed-Loop Reset Control System as Shown in Fig. \ref{fig: CL_Block_Diagram_Harmonics_a}, Decomposing \(v(t)\) into Its \(n\)-th Harmonic Components, Expressed as \(v(t) = \sum\nolimits_{n=1}^{\infty} v^n(t)\).}

\begin{figure}[htp]
    \centering
    % \captionsetup{singlelinecheck = false, format= hang, justification=justified, font=footnotesize, labelsep=space}
    \includegraphics[width=0.75\columnwidth]{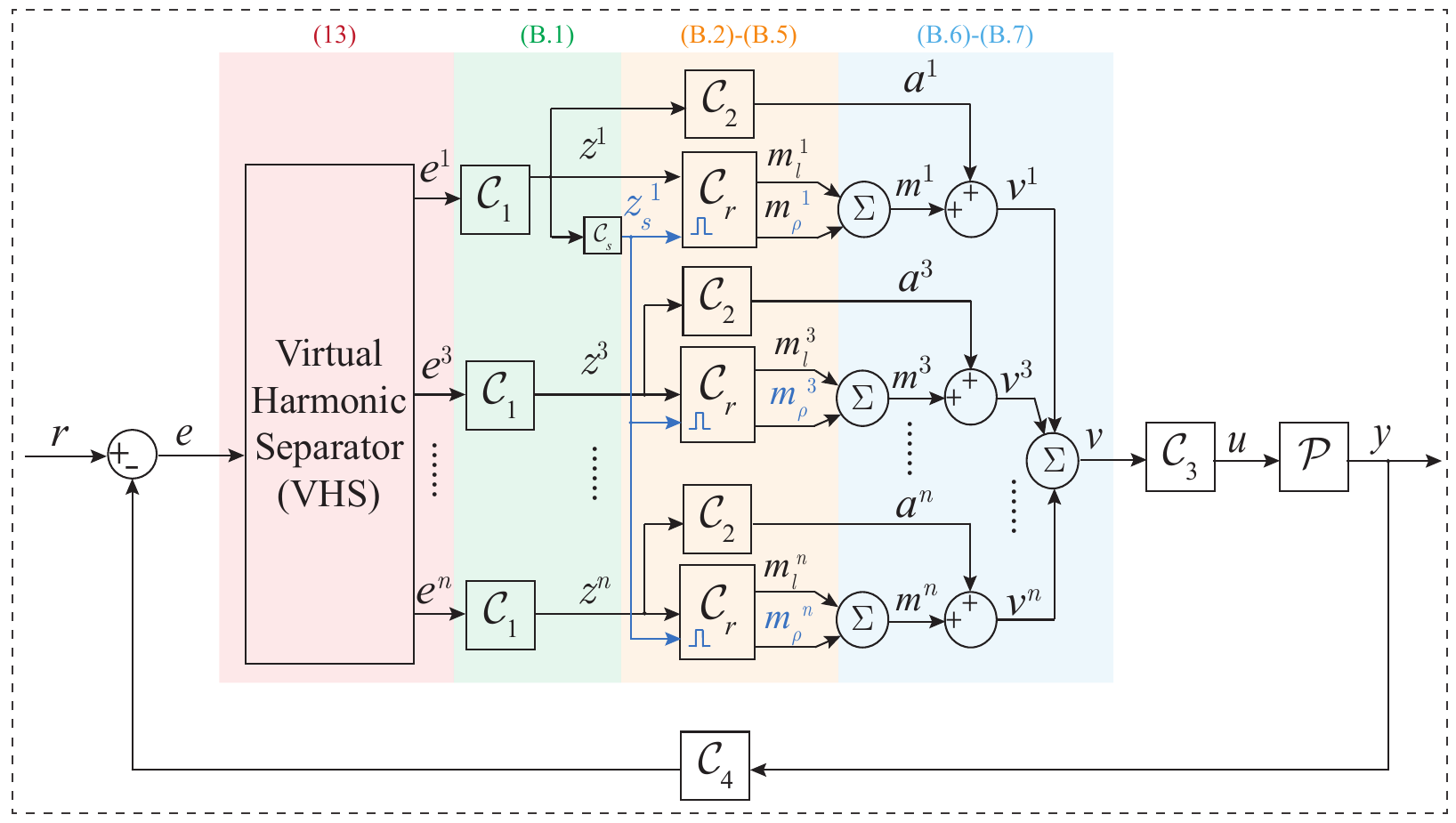}
    \caption{Block diagram of the closed-loop reset control system, showing the decomposition of \(v(t)\) into its \(n\)-th harmonic components, expressed as \(v(t) = \sum\nolimits_{n=1}^{\infty} v^n(t)\). The colored blocks correspond to the equations with matching colors.}
	\label{fig: CL_Block_Diagram_Harmonics_a}
\end{figure}
% The following content constructs the block diagram of the closed-loop reset control system shown in 
Figure \ref{fig: CL_Block_Diagram_Harmonics_a} illustrates the signal flow from the error signal \( e(t) \) to the reset controller's output \( v(t) \) in Fig. \ref{fig1:RC system}. First, the error signal \( e(t) \) passes through the ``Virtual Harmonic Separator" and \(\mathcal{C}_1\), resulting in the signal \( z^n(t) \). Subsequently, \( z^n(t) \) is processed by \(\mathcal{C}_2\) to generate \( a^n(t) \), and by \(\mathcal{C}_r\) to produce \( m^n(t) \). The signals \( a^n(t) \) and \( m^n(t) \) are then combined to form \( v^n(t) \), and summing all \( v^n(t) \) components produces the overall output \( v(t) \).

Due to the nonlinearity of the reset control system, the signal \( e(t) \) contains an infinite series of harmonics \( e^n(t) \) for \( n = 2k + 1, \, k \in \mathbb{N} \), as formulated in \eqref{eq: e,y,u}. Utilizing the ``Virtual Harmonic Separator," each harmonic component \( e^n(t) \) is individually separated, as illustrated in Fig. \ref{fig: CL_Block_Diagram_Harmonics_a}. Subsequently, each separated harmonic \( e^n(t) \) is filtered through the linear transfer function \( \mathcal{C}_1(n\omega) \), producing the corresponding output signal \( z^n(t) \) given by
\begin{equation}
\label{eq: wn(t)1}
    \begin{aligned}
z^n(t) &= |Z^n |\sin(n\omega t+ \angle  Z^n), \\
        |Z^n | &= |E^n\mathcal{C}_{1}(n\omega)|,\\
        \angle Z^n &= \angle E^n + \angle \mathcal{C}_{1}(n\omega).
    \end{aligned}
\end{equation}
Signals \(z^n(t)\), processed through the blocks \(\mathcal{C}_r\) and \(\mathcal{C}_2\), produce outputs \(m^n(t)\) and \(a^n(t)\), respectively. We first derive the expression for \(m^n(t)\).

% Define \(\mathcal{C}_r^n(\omega)\) as the transfer function from \(z^n(t)\) to \(m^n(t)\), represented as the block \(\mathcal{C}_r^n\) in Fig. \ref{fig: CL_Block_Diagram_Harmonics_a}. 
The reset controller \(\mathcal{C}_r\) processes the input signal \(z^n(t)\) from \eqref{eq: wn(t)1} and the reset trigger signal \(z_s^1(t)\), given by
\begin{equation}
\label{eq: z_s1(t)}
    \begin{aligned}
z_s^1(t) &= |Z_s^1|\sin(\omega t+\angle Z_s^1),\\ 
|Z_s^1| &= |{Z}^1(\omega)|\cdot|\mathcal{C}_{s}(\omega)|= |E^1\mathcal{C}_{1}(\omega)\mathcal{C}_{s}(\omega)|,\\
\angle Z_s^1&=  \angle Z^1 + \angle \mathcal{C}_{s}(\omega) = \angle E^1 + \angle \mathcal{C}_{1}(\omega)+ \angle \mathcal{C}_{s}(\omega).
    \end{aligned}
\end{equation}
Referring \cite{ZHANG2024106063}, a reset controller $\mathcal{C}_r$ with the input signal $z^n(t)$ in \eqref{eq: wn(t)1} and the reset triggered signal $z_s^1(t)$ in \eqref{eq: z_s1(t)} generates output $m^n(t)$ given by
\begin{equation}
\label{eq: mn(t)}
    m^n(t) = m^n_{l}(t) + m^n_{\rho}(t),
\end{equation}
where $m^n_{l}(t)$ is the linear component, given by
\begin{equation}
\label{eq: mn_{l}(t)}
\begin{aligned}
 m^n_{l}(t) & = |M^n_{l}|\sin(n\omega t+ \angle M^n_{l}), \\
 |M^n_{l}| &= |Z^n\mathcal{C}_{l}(n\omega)|= |E^n\mathcal{C}_{1}(n\omega)\mathcal{C}_{l}(n\omega)|,\\
 \angle M^n_{l} &= \angle Z^n+\angle \mathcal{C}_{l}(n\omega)=\angle E^n + \angle \mathcal{C}_{1}(n\omega)+\angle \mathcal{C}_{l}(n\omega),
\end{aligned}  
\end{equation}
and $m^n_{\rho}(t)$ is the nonlinear component, given by
\begin{equation}
\label{eq: m_rho_{l}(t)}
\begin{aligned}
m_{\rho}^n(t) &= \sum\nolimits_{\eta=1}^{\infty}\mathscr{F}^{-1}[\Delta_x(\eta\omega)Q_\eta(\omega)],\ \eta=2k+1,\ k\in\mathbb{N},\\
\Delta_q^n(\omega) &= (I+e^{A_R\pi/\omega})(A_\rho e^{A_R\pi/\omega}+I)^{-1}(I-A_\rho)\Delta_c^n(\omega),\\ 
% m_{l}(t) &= |Z^n \mathcal{C}_{l}(n\omega)| \sin(n\omega t + \angle Z^n + \angle \mathcal{C}_{l}(n\omega)),\\
Q_\eta(\omega) &= 2|Z^n|\Delta_q^n(\omega) \mathscr{F}[\sin(\eta\omega t + \eta\angle Z_s^1)]/(\eta\pi),\\
\Delta_c^n(\omega) &= |\Delta_l(n\omega)|\sin(\angle\Delta_l(n\omega)+ \angle Z^n-n\angle Z_s^1),\\
\Delta_x(\eta\omega) &= C_R(j\eta\omega I-A_R)^{-1}j\eta\omega I,\\
\Delta_l(n\omega) &= (jn\omega I-A_R)^{-1}B_R .       
    \end{aligned}
\end{equation}
Meanwhile, the LTI system \( \mathcal{C}_2\) processes the input signal \( z^n(t) \) in \eqref{eq: wn(t)1} and produces the output \( a^n(t) \), given by
\begin{equation}
\label{eq: an(t)}
\begin{aligned}
a^n(t) &= |A^n|\sin(n\omega t+ \angle A^n), \\ 
|A^n| &= |Z^n|\cdot|\mathcal{C}_{2}(n\omega)| = |E^n\mathcal{C}_{1}(n\omega)\mathcal{C}_{2}(n\omega)|,\\
\angle A^n &=  \angle Z^n+ \angle \mathcal{C}_{2}(n\omega) = \angle E^n + \angle \mathcal{C}_{1}(n\omega)+ \angle \mathcal{C}_{2}(n\omega).  
\end{aligned}
\end{equation}
By summing \( m^n(t) \) from \eqref{eq: mn(t)} and \( a^n(t) \) from \eqref{eq: an(t)}, the signal \( v(t) \) in Fig. \ref{fig: CL_Block_Diagram_Harmonics_a}, is obtained as follows:
\begin{equation}
\label{eq: vt_1}
\begin{aligned}
v(t) &=  \sum\nolimits_{n=1}^{\infty} v^n(t), \\
v^n(t) &=  m^n(t) +  a^n(t).
\end{aligned}     
\end{equation}
Here, the block diagram in Fig. \ref{fig: CL_Block_Diagram_Harmonics_a} is constructed. Next, based on Fig. \ref{fig: CL_Block_Diagram_Harmonics_a}, the block diagram in Fig. \ref{fig: CL_Block_Diagram_Harmonics_b} is developed.\\
% Building on this, the next section presents the block diagram in Fig. \ref{fig: CL_Block_Diagram_Harmonics_b}, which decomposes the signal \( v(t) \) into its linear component \( v_l(t) \) and nonlinear component \( v_\rho(t) \).\\
% The corresponding analytical expressions for these components are then presented.
\textbf{Step 2: Build the Block Diagram for Analyzing the Sinusoidal-Input Closed-Loop Reset Control System as Shown in Fig. \ref{fig: CL_Block_Diagram_Harmonics_b}, Decomposing \(v(t)\) into a Linear Component \(v_l(t)\) and a Nonlinear Component \(v_\rho(t)\).}

\begin{figure}[htp]
    \centering
    % \captionsetup{singlelinecheck = false, format= hang, justification=justified, font=footnotesize, labelsep=space}
    \includegraphics[width=0.75\columnwidth]{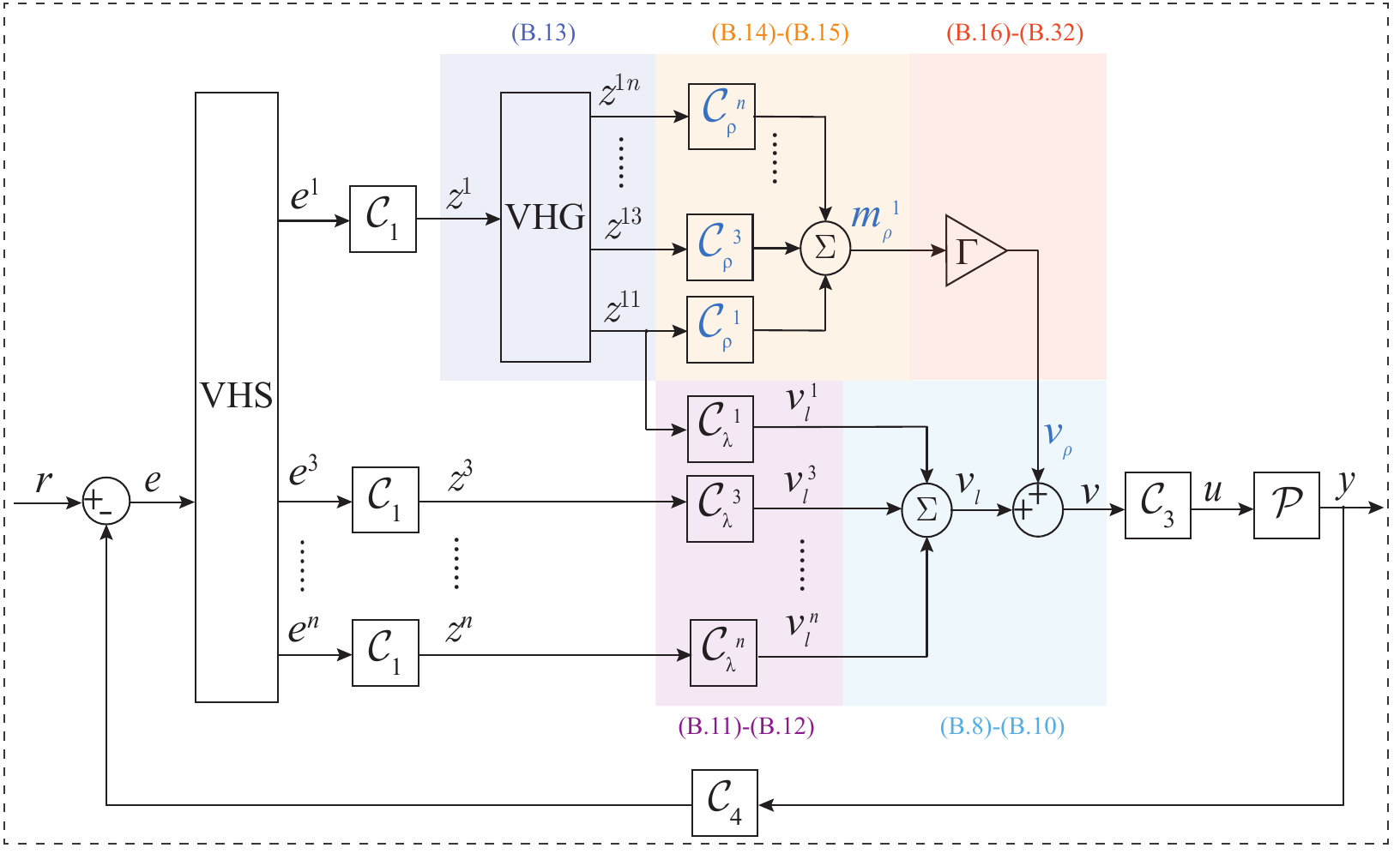}
    \caption{Block diagram of the closed-loop reset control system, illustrating the decomposition of \(v(t)\) into its linear component \(v_l(t)\) and nonlinear component \(v_\rho(t)\).}
	\label{fig: CL_Block_Diagram_Harmonics_b}
\end{figure}
From \eqref{eq: mn(t)} and \eqref{eq: vt_1}, $v^n(t)$ is given by
\begin{equation}
\label{eq: v^n2}
v^n(t) =  m_l^n(t) + m_\rho^n(t) +  a^n(t).    
\end{equation}
From \eqref{eq: vt_1} and \eqref{eq: v^n2}, signal $v(t)$ can be expressed as\begin{equation}
\label{eq: vln(t)1}
v(t) =  v_l(t) +  v_\rho(t),
\end{equation}
where
\begin{equation}
\label{eq: v_ln(t)}
\begin{aligned}
v_\rho(t) &= \sum\nolimits_{n=1}^{\infty} m_\rho^n(t),\\
v_l(t) &= \sum\nolimits_{n=1}^{\infty} v_l^n(t),\\
v_l^n(t) &=  m_l^n(t) +  a^n(t).    
\end{aligned}
\end{equation}
Define $Z^n(\omega)$, $V(\omega)$, $V^n(\omega)$, $V_l(\omega)$, $V_\rho(\omega)$, $M_\rho^n(\omega)$, $M_l^n(\omega)$, and $A^n(\omega)$ as the Fourier transforms of $z^n(t)$, $v(t)$, $v^n(t)$, $v_l(t)$, $v_\rho(t)$, $m_\rho^n(t)$, $m_l^n(t)$, and $a^n(t)$, respectively. 
% From \eqref{eq: vt_1}, \eqref{eq: vln(t)1}, and \eqref{eq: v_ln(t)}, $V^{n}(\omega)$ is given by
% \begin{equation}
% \label{eq: Vn0}
% V^n(\omega) =   M^n_{l}(\omega) +  A^n(\omega) + M^n_{\rho}(\omega).  
% \end{equation}
From \eqref{eq: mn_{l}(t)}, \eqref{eq: an(t)}, and \eqref{eq: v_ln(t)}, $V_l(\omega)$ is expressed as
\begin{equation}
V_l(\omega) = \sum\nolimits_{n=1}^{\infty} V_l^n(\omega),
\end{equation}
where 
\begin{equation}
\label{eq: v_ln(w)}
\begin{aligned}
V_l^n(\omega) &= Z^n(\omega) \mathcal{C}_\lambda^n(\omega),\\
 \mathcal{C}_\lambda^n(\omega) &= \mathcal{C}_l(n\omega) + \mathcal{C}_2(n\omega).  
\end{aligned} 
\end{equation}
From \eqref{eq: vln(t)1}, we have \( V(\omega) = V_l(\omega) + V_\rho(\omega) \). The function \( V_l(\omega) \) is derived in \eqref{eq: v_ln(w)}. The following content derives \( V_\rho(\omega) \).

Consider the reset controller $\mathcal{C}_r$ in Fig. \ref{fig: CL_Block_Diagram_Harmonics_a} with the input signal $z^1(t) = |Z^1|\sin(\omega t+ \angle Z^1)$ and the reset triggered signal $z_s^1(t) = |Z^1\mathcal{C}_s(\omega)|\sin(\omega t+ \angle Z^1 +\angle \mathcal{C}_s(\omega))$. Using the ``Virtual Harmonics Generator", the input signal $z^1(t)$ generates $n$ harmonics, defined as
\begin{equation}
\label{eq:z1n_def}
 z^{1n}(t) =   |Z^1|\sin(n\omega t+ n\angle Z^1), 
\end{equation}
whose Fourier transform is $Z^{1n}(\omega) = \mathscr{F}[z^{1n}(t)]$. 

According to Theorem \ref{thm: open-loop HOSIDF}, each harmonic $ z^{1n}(t)$ filtered through the reset controller $\mathcal{C}_r$ generates two components: $m_l^1(t)$ in \eqref{eq: mn_{l}(t)} and the nonlinear output signal $m^1_{\rho}(t)$, which is given by
\begin{equation}
\begin{aligned}
\label{eq: w_rho1(t), w_1n(t)}
% m^1(t) &= m_l^1(t) + m_\rho^1(t),\\
m_\rho^1(t) &= \sum\nolimits_{n=1}^{\infty}|Z^1\mathcal{C}_\rho^n(\omega)|\sin(n\omega t+ n\angle Z^1 + \angle\mathcal{C}_\rho^n(\omega)).
\end{aligned}  
\end{equation}
From \eqref{eq: w_rho1(t), w_1n(t)}, the Fourier transform of $m_\rho^1(t)$ is given by
\begin{equation}
\begin{aligned}
\label{eq: w_rho1(w), w_1n(w)}
M_\rho^1(\omega) &= \sum\nolimits_{n=1}^{\infty}Z^{1n}(\omega)\mathcal{C}_\rho^n(n\omega).
\end{aligned}      
\end{equation}
As shown in Fig. \ref{fig: CL_Block_Diagram_Harmonics_a}, the reset controller \( \mathcal{C}_r \) processes distinct input signals \( z^n(t) \) while relying on the same reset-triggered signal \( z_s^1(t) \). Then, define \( v_\rho(t) = \sum\nolimits_{n=1}^{\infty} m^n_{\rho}(t) \) as the summation of nonlinear components \( m^n_{\rho}(t) \). From \eqref{eq: m_rho_{l}(t)}, \( V_\rho(\omega) = \mathscr{F}[v_\rho(t)] \) is given by:
\begin{equation}
\label{eq: M_rho_n(w)}
    \begin{aligned}
     V_{\rho}(\omega) &=   \sum\nolimits_{n=1}^{\infty}M^{n}_{\rho}(\omega), \\
     M^{n}_{\rho}(\omega)&= \sum\nolimits_{\eta=1}^{\infty} \Delta_x(\eta\omega)Q_\eta(\omega),
    \end{aligned}
\end{equation}
where $\Delta_x(\eta\omega)$ and $Q_\eta(\omega)$ are given in \eqref{eq: m_rho_{l}(t)}. 

Then, a factor \(\Gamma(\omega)\) is introduced to represent the ratio of \(V_{\rho}(\omega)\) to \(M^1_{\rho}(\omega)\), defined as:
\begin{equation}
\label{eq: G0}
\Gamma(\omega)=\frac{V_{\rho}(\omega)}{M^{1}_{\rho}(\omega)} =\frac{\sum\nolimits_{n=1}^{\infty}M^{n}_{\rho}(\omega)}{M^{1}_{\rho}(\omega)}, \text{ where }n=2k+1,k\in\mathbb{N}.
\end{equation}
According to \eqref{eq: m_rho_{l}(t)}, the nonlinear components \( m^n_{\rho}(t) \) for different \(n\) share identical phase and period. Therefore, $\Gamma(\omega)$ is a real number.

% The following derivation from equations \eqref{eq: vnl} to \eqref{eq: gamma_final} analytically determines the value of \(\Gamma(\omega)\). 
% with their corresponding Fourier transforms denoted as \( V_\rho(\omega) \) and \( M_\rho^n(\omega) \).
From \eqref{eq: M_rho_n(w)} and \eqref{eq: G0}, $\Gamma(\omega)$ is expressed as
\begin{equation}
\label{eq: G2}
\Gamma(\omega)=\frac{\sum\nolimits_{n=1}^{\infty}|Z^n|\Delta_c^n(\omega)}{|Z^1|\Delta_c^1(\omega)},
\end{equation}
where $\Delta_c^n(\omega)$ from \eqref{eq: m_rho_{l}(t)} is given by
\begin{equation}
\label{Delta_cn2}
\Delta_c^n(\omega) =|\Delta_l(n\omega)| \sin(\angle \Delta_l(n\omega)+ \angle Z^n -n\angle Z_s^1 ).
\end{equation}
However, in \eqref{eq: G2} and \eqref{Delta_cn2}, the ratio \( {|Z^n|}/{|Z^1|} \) and the value of \( (\angle Z^n - n\angle Z_s^1) \) remain undetermined. Therefore, the subsequent analysis addresses these unknown parameters by leveraging the underlying system dynamics and harmonic relationships, enabling the determination of \( \Gamma(\omega) \).

From \eqref{eq: w_rho1(w), w_1n(w)} and \eqref{eq: G0}, $V_{\rho}(\omega)$ can be expressed as
\begin{equation}
\label{eq: vnl}
    \begin{aligned}
V_{\rho}(\omega) &= \sum\nolimits_{n=1}^{\infty}\Gamma(\omega)Z^{1n}(\omega)\mathcal{C}_{\rho}^n(\omega). 
    \end{aligned}
\end{equation}
Let the \( n \)th harmonic component in \( V_\rho(\omega) \) from \eqref{eq: M_rho_n(w)} and \eqref{eq: vnl} be equal to each other. Then, \( M_{\rho}^n(\omega) \) is determined as
\begin{equation}
\label{eq: vnl_n}
M_{\rho}^n(\omega) = \Gamma(\omega)Z^{1n}(\omega)\mathcal{C}_{\rho}^n(\omega).  
\end{equation}
From \eqref{eq: vln(t)1}, \eqref{eq: v_ln(w)}, and \eqref{eq: vnl_n}, $V^n(\omega)$ is given by
\begin{equation}
\label{eq: Un}
\begin{aligned}
V^n(\omega) = Z^n(\omega)\mathcal{C}_\lambda^n(\omega) + \Gamma(\omega)Z^{1n}(\omega)\mathcal{C}_{\rho}^n(\omega).
\end{aligned}
\end{equation}
In the closed-loop system in Fig. \ref{fig: CL_Block_Diagram_Harmonics_a}, $Z^n(\omega)$ is given by
\begin{equation}
\label{eq:E_n}
    Z^n(\omega) = - V^n(\omega)\mathcal{C}_3(n\omega)\mathcal{P}(n\omega)\mathcal{C}_4(n\omega)\mathcal{C}_1(n\omega).
\end{equation}
Substituting \eqref{eq: Un} into \eqref{eq:E_n}, we have
\begin{equation}
\label{eq: En2}
    Z^n(\omega) = - Z^n(\omega)\mathcal{L}_{l}(n\omega) - \Gamma(\omega)Z^{1n}(\omega)\mathcal{L}_{\rho}(n\omega),
\end{equation}
where
\begin{equation}
\label{eq:Lbl,Lnl}
    \begin{aligned}
    \mathcal{L}_{l}(n\omega) &= {\mathcal{C}_\lambda^n(\omega)\mathcal{C}_{3}(n\omega)\mathcal{P}(n\omega)\mathcal{C}_{4}(n\omega)\mathcal{C}_{1}(n\omega)},\\
    \mathcal{L}_{\rho}(n\omega) &= \mathcal{C}_{\rho}^n(\omega)\mathcal{C}_3(n\omega)\mathcal{P}(n\omega)\mathcal{C}_4(n\omega)\mathcal{C}_1(n\omega).
    \end{aligned}
\end{equation}
% Based on the definitions of $z^n(t)$ in \eqref{eq: e,y,u} and $z^{1n}(t)$ in \eqref{eq:z1n_def}, equation \eqref{eq: En2} is transferred into the time-domain expression as follows:
% \begin{equation}
% \label{eq:phase and Mag}
% \begin{aligned}
% &|Z^n|\cdot|1+\mathcal{L}_{l}(n\omega)| \sin(n\omega t + \angle Z^n + \angle (1+\mathcal{L}_{l}(n\omega))) =
% % \\&
% -\Gamma(\omega)|Z^1|\cdot|\mathcal{L}_{\rho}(n\omega)| \sin(n\omega t + n\angle Z^1 + \angle\mathcal{L}_{\rho}(n\omega)).
% \end{aligned}
% \end{equation}
% From \eqref{eq:phase and Mag} and the given condition where $|Z^n|>0$, 
From \eqref{eq: wn(t)1}, \eqref{eq:z1n_def}, the relation between $Z^{1n}(\omega)$ and $Z^n(\omega)$ is given by
\begin{equation}
\label{eq:z_1n and zn}
Z^{1n}(\omega) = \frac{|Z^1(\omega)|  e^{j(n\angle Z^1(\omega)-\angle Z^{n}(\omega))}}{|Z^{n}(\omega)|}Z^{n}(\omega).   
\end{equation}
From \eqref{eq: En2} and \eqref{eq:z_1n and zn}, the following equations can be deduced:
\begin{equation}
\label{eq: phase and Mag2}
\begin{cases}
|Z^n| =  \Gamma(\omega)\Psi_n(\omega)|Z^1|, &\text{ for } n>1\\
\angle Z^n= n\pi+ n\angle Z^1 + \angle\mathcal{L}_{\rho}(n\omega) - \angle (1+\mathcal{L}_{l}(n\omega)),  &\text{ for } n>1,
\end{cases}
\end{equation}
where
\begin{equation}
\label{eta}
    \begin{aligned}        
	\Psi_n(\omega) &= {|\mathcal{L}_{\rho}(n\omega)|}/{|1+\mathcal{L}_{l}(n\omega)|}.
    \end{aligned}
\end{equation}
By substituting the phase relationship between \( \angle Z^1 \) and \( \angle Z^n \) (\(n > 1\)) from \eqref{eq: phase and Mag2} and \( \angle Z_s^1 = \angle Z^1 + \angle \mathcal{C}_s(\omega) \) from \eqref{eq: z_s1(t)} into \eqref{Delta_cn2}, we derive:
\begin{enumerate}
    \item For $n=1$,\\
    \begin{equation}
    \label{Delta_cn31}
     \Delta_c^1(\omega) =   |\Delta_l(\omega)| \sin(\angle \Delta_l(\omega)-\angle \mathcal{C}_s(\omega)).
     \end{equation}
     \item  For $n=2k+1>1$,\\
     \begin{equation}
    \label{Delta_cn32}
     \begin{aligned}
     \Delta_c^n(\omega)& = |\Delta_l(n\omega)|\sin(\angle \Delta_l(\omega)+ \angle Z^n -\angle Z_s^1)
     \\ &
     =-|\Delta_l(n\omega)|\sin(\angle \Delta_l(n\omega)+\angle\mathcal{L}_{\rho}(n\omega)- \angle (1+\mathcal{L}_{l}(n\omega))-n\angle \mathcal{C}_s(\omega)). 
     \end{aligned}
     \end{equation}
\end{enumerate}
Substituting $\Delta_c^n(\omega)$ from \eqref{Delta_cn31} and \eqref{Delta_cn32} into \eqref{eq: G2}, we obtain
\begin{equation}
\label{eq: G4}
\Gamma(\omega)=1+\frac{\sum\nolimits_{n=3}^{\infty}|Z^n|\Delta_c^n(\omega)}{|Z^1|\Delta_c^1(\omega)}.
\end{equation}
Then, by substituting $|Z^n| =  \Gamma(\omega)\Psi_n(\omega)|Z^1|$ from \eqref{eq: phase and Mag2} into \eqref{eq: G4}, $\Gamma(\omega)$ is derived as
% \begin{equation}
% \label{eq: G5}
% \Gamma(\omega)=1+\Gamma(\omega)\frac{\sum\nolimits_{n=3}^{\infty}\Psi_n(\omega)\Delta_c^n(\omega)}{\Delta_c^1(\omega)}.
% \end{equation}
% % where $\Delta_c(\omega)=|\Delta_l(\omega)| \sin(\angle \Delta_l(\omega))$ is defined in \eqref{Cnl_final}. 
% From \eqref{eq: G5}, $\Gamma(\omega)$ is obtained as below:
\begin{equation}
		\label{eq: gamma_final}		
		\begin{aligned}
		  \Gamma(\omega) &= 1/\left(1-{\sum\nolimits_{n=3}^{\infty}\Psi_n(\omega)\Delta_c^n(\omega)}/{\Delta_c^1(\omega)}\right), n = 2k+1, k\in\mathbb{N}.
		\end{aligned}
	\end{equation}
Up to this point, the block diagram of the closed-loop reset control system in Fig. \ref{fig: CL_Block_Diagram_Harmonics_b} has been constructed. Based on this, the next step derives the HOSIDFs for the closed-loop reset control systems. \\
\textbf{Step 3: Derive the HOSIDFs for Closed-Loop Reset Control Systems.}
% From Fig. \ref{fig: CL_Block_Diagram_Harmonics_b}, the first-order sensitivity function \(\mathcal{S}_1(\omega)\) for the closed-loop reset control system is derived.

By applying the ``Virtual Harmonics Generator", the input signal $r(t)=|R|\sin(\omega t)$ generates $n$ harmonics, defined as
\begin{equation}
\label{eq:rn(t)_def}
    r^n(t) = |R|\sin(n\omega t),
\end{equation}
whose Fourier transform is denoted as $R^n(\omega)$.

The output signal $y(t)$ of the closed-loop reset control system includes infinite many harmonics $y_n(t)$, as defined in \eqref{eq: e,y,u}. Define $Y^n(\omega) = \mathscr{F}[y_n(t)]$. From the block diagram in Fig. \ref{fig: CL_Block_Diagram_Harmonics_b}, we have
     \begin{equation}
     \label{eq:Y1(w)}
         Y^1(\omega) = E^1(\omega)[\mathcal{L}_{l}(\omega) + \Gamma(\omega)\mathcal{L}_{\rho}(\omega)]/\mathcal{C}_{4}(\omega).
     \end{equation}
% where
% \begin{equation}
% \label{eq: C_psi, C_phi}
% \begin{aligned}
% \mathcal{C}_{\psi}(\omega) &= \mathcal{C}_{1}(\omega)\mathcal{C}_\rho^1(\omega)\mathcal{C}_{3}(\omega)\mathcal{P}(\omega),\\   
% \mathcal{C}_{\phi}(\omega) &= \mathcal{C}_{1}(\omega)\mathcal{C}_\lambda^1(\omega)\mathcal{C}_3(\omega)\mathcal{P}(\omega).    
% \end{aligned}   
% \end{equation}
In the closed loop, the following relation holds:
\begin{equation}
\label{eq:E1(w)}
    E^{1}(\omega) = R^1(\omega)-\mathcal{C}_{4}(\omega)Y^{1}(\omega).
\end{equation}
Combining \eqref{eq:Y1(w)} and \eqref{eq:E1(w)}, the first-order sensitivity function $\mathcal{S}_1(\omega)$ for the closed-loop reset control system is defined as 
\begin{equation}
\label{eq: S1}
\mathcal{S}_1(\omega)= \frac{E^1(\omega)}{R^1(\omega)} = \frac{1}{1+\mathcal{L}_{o}(\omega)},
% \frac{1}{1+\mathcal{L}_{l}(\omega)+\Gamma(\omega)\mathcal{L}_{\rho}(\omega)},
\end{equation}
where
\begin{equation}
\label{eq:Ll, L_rho2}
    \begin{aligned}
        % \mathcal{L}_{l}(\omega) &= \mathcal{C}_{\phi}(\omega)\mathcal{C}_{4}(\omega),\\
        % \mathcal{L}_{\rho}(\omega) &= \mathcal{C}_{\psi}(\omega)\mathcal{C}_{4}(\omega),\\
	\mathcal{L}_{o}(n\omega) &= \mathcal{L}_{l}(n\omega)+\Gamma(\omega)\mathcal{L}_{\rho}(n\omega) =  \mathcal{L}_{n}(\omega)+(\Gamma(\omega)-1)\mathcal{L}_{\rho}(n\omega).        
    \end{aligned}
\end{equation}
% \begin{equation}
%     \label{eq: Lo}
% 	\mathcal{L}_{o}(n\omega) = \mathcal{L}_{l}(n\omega)+\Gamma(\omega)\mathcal{L}_{\rho}(n\omega),
% \end{equation}

The subsequent content focuses on deriving the higher-order sensitivity function \(\mathcal{S}_n(\omega)\) for \(n > 1\) for the closed-loop reset control system.

% Define $Z^n(\omega)$ and $W(\omega)$ as the Fourier transforms of signals $z^n(t)$ and $w(t)$, respectively. From \eqref{eq: wn(t)1} and \eqref{eq: S1}, the transfer function from $R^1(\omega)$ to $Z^1(\omega)$ defined as $\mathcal{C}_z^1(\omega)$ is given by 
% \begin{equation}
% \label{eq: C_w1(w)}
%     \mathcal{C}_z^1(\omega) = \frac{Z^1(\omega)}{R^1(\omega)} = \frac{\mathcal{C}_1(\omega)}{1+\mathcal{L}_{o}(\omega)}.
% \end{equation}
From \eqref{eq:z1n_def}, $Z^{1n}(\omega)= \mathscr{F}[z^{1n}(t)]$ is expressed as
\begin{equation}
\label{eq: W^1n(w)}
    Z^{1n}(\omega) = |\mathcal{C}_1(\omega)\mathcal{S}_1(\omega)|R^n(\omega)e^{jn(\angle \mathcal{C}_1(\omega)+\angle \mathcal{S}_1(\omega))}.
\end{equation}
From the block diagram in Fig. \ref{fig: CL_Block_Diagram_Harmonics_b}, the $n$th order harmonic $Z^n(\omega)$ is given by
\begin{equation}
\label{eq: W^n(w)1}
Z^n(\omega) = -Z^n(\omega)\mathcal{L}_{l}(n\omega) - \Gamma(\omega)Z^{1n}(\omega)\mathcal{L}_{\rho}(n\omega).
\end{equation}
Substituting $Z^{1n}(\omega)$ from \eqref{eq: W^1n(w)} into \eqref{eq: W^n(w)1}, we have
\begin{equation}
\begin{aligned}
\label{eq: W^n(w)2}
    Z^n(\omega) = -Z^n&(\omega)\mathcal{L}_{l}(n\omega) - \Gamma(\omega)|\mathcal{C}_1(\omega)\mathcal{S}_1(\omega)|
    % \cdot\\    &
    R^n(\omega)e^{jn(\angle \mathcal{C}_1(\omega)+\angle \mathcal{S}_1(\omega))}\mathcal{L}_{\rho}(n\omega).
\end{aligned}   
\end{equation}
From \eqref{eq: W^n(w)2}, we obtain:
\begin{equation}
\label{C_wn}
\begin{aligned}
\frac{Z^n(\omega)}{R^n(\omega)}= -\mathcal{S}_{l}(n\omega)\cdot \Gamma(\omega)|\mathcal{C}_1(\omega)\mathcal{S}_1(\omega)|e^{jn(\angle \mathcal{C}_1(\omega)+\angle \mathcal{S}_1(\omega))}\mathcal{L}_{\rho}(n\omega),
\end{aligned} 
\end{equation}
where $\mathcal{S}_{l}(n\omega)$ denotes the sensitivity function of the BLS, given by
\begin{equation}
    \mathcal{S}_{l}(n\omega) =  \frac{1}{1+\mathcal{L}_{l}(n\omega)}.
\end{equation}
From \eqref{eq: wn(t)1}, the relationship between $Z^n(\omega)$ and ${E}^n(\omega)$ is given by
\begin{equation}
\label{eq:Zn, En, C1}
Z^n(\omega) = {E}^n(\omega)\mathcal{C}_1(n\omega).
\end{equation}
From \eqref{eq: H_hn}, \eqref{C_wn}, and \eqref{eq:Zn, En, C1}, the $n$th order (for $n>1$) harmonic in the sensitivity function for the closed-loop reset control system is given by
\begin{equation}
\label{eq:Sn1}
\begin{aligned}
   \mathcal{S}_n(\omega) &= \frac{{E}^n(\omega)}{R^n(\omega)} = \frac{Z^n(\omega)}{R^n(\omega)\mathcal{C}_1(n\omega)}
   % \\    &
   = - \mathcal{S}_{l}(n\omega) \cdot|\mathcal{S}_1(\omega)|e^{jn\angle \mathcal{S}_1(\omega)} \cdot\Gamma(\omega)\mathcal{L}_n(\omega)\mathcal{C}_4(n\omega).
\end{aligned}    
\end{equation}
% where 
% \begin{equation}
% \label{eq: L_alpha}
% \mathcal{L}_{\alpha}(n\omega) = \mathcal{C}_{\rho}^n(\omega)\mathcal{C}_{3}(n\omega)\mathcal{P}(n\omega)\mathcal{C}_{4}(n\omega).    
% \end{equation}
The $n$th order complementary sensitivity function $\mathcal{T}_n(\omega)$ and the control sensitivity function $\mathcal{CS}_n(\omega)$ can be derived through a same procedure as $\mathcal{S}_n(\omega)$ from \eqref{eq:Y1(w)} to \eqref{eq:Sn1}. Here, we concludes the proof.
\end{proof}

\end{document}